\documentclass[12pt]{article}

\usepackage{comment}
\usepackage{ragged2e}
\usepackage{graphicx}
\usepackage{IEEEtrantools}
\usepackage{amsmath}
\usepackage{amssymb}
\usepackage{multirow,array}

\usepackage{float} 
\usepackage{tablefootnote}

\usepackage{subcaption}
\usepackage{amssymb,amsmath,amsfonts,eurosym,geometry,ulem,graphicx,caption,color,setspace,sectsty,comment,footmisc,caption,natbib,pdflscape,array,hyperref,bbm}
\usepackage{graphicx}

 \usepackage{mathtools}
\newtheorem{claim}{Claim}
\newtheorem{theorem}{Theorem}
\newtheorem{corollary}[theorem]{Corollary}
\newtheorem{proposition}{Proposition}

\newtheorem{assumption}{Assumption}

\newtheorem{definition}{Definition}

\newenvironment{proof}[1][Proof]{\noindent\textbf{#1:} }{\ \rule{0.5em}{0.5em}}
\newenvironment{example}[1][Example]{\noindent\textbf{#1} }{\ \rule{0.5em}{0.5em}}

 \usepackage[flushleft]{threeparttable}

\newcolumntype{L}[1]{>{\raggedright\let\newline\\arraybackslash\hspace{0pt}}m{#1}}
\newcolumntype{C}[1]{>{\centering\let\newline\\arraybackslash\hspace{0pt}}m{#1}}
\newcolumntype{R}[1]{>{\raggedleft\let\newline\\arraybackslash\hspace{0pt}}m{#1}}

\usepackage{natbib}
\begin{document}

\title{\Large Experience-weighted attraction learning in network coordination games}
\author{Fulin Guo }

\date{}
\maketitle
\begin{abstract}

This paper studies the action dynamics of network coordination games with bounded-rational agents. I apply the experience-weighted attraction (EWA) model to the analysis as the EWA model has several free parameters that can capture different aspects of agents' behavioural features. I show that the set of possible long-term action patterns can be largely different when the behavioural parameters vary, ranging from a unique possibility in which all agents favour the risk-dominant option to some set of outcomes richer than the collection of Nash equilibria. Monotonicity and non-monotonicity in the relationship between the number of possible long-term action profiles and the behavioural parameters are explored. I also study the question of influential agents in terms of whose initial predispositions are important to the actions of the whole network. The importance of agents can be represented by a left eigenvector of a Jacobian matrix provided that agents' initial attractions are close to some neutral level. Numerical calculations examine the predictive power of the eigenvector for the long-run action profile and how agents' influences are impacted by their behavioural features and network positions.

\end{abstract}

\newpage

\section{Introduction}
\label{sec:introduction}
The study of how the interactions of self-interested individuals generate aggregate social outcomes has been an important topic in economics. This paper further explores this theme in the context of network games,\footnote{See, e.g., Chapter 4 of \cite{goyal2023} and Chapter 3 of \cite{goyal2009connections} for a review of network games.} which studies strategic interactions of agents when they are in network environments. Most of the models of network games assume that agents are rational and typically focus on the characterisation of Nash equilibrium (or its variants) as in classical game theory. While this sort of analysis provides benchmarks for the understanding of social outcomes under network interactions, it is important to further study how bounded-rational agents will play the network games for the following two reasons. First, network environments are usually complex, so it may not be realistic to assume that agents will play a Nash equilibrium right away, which indicates the importance of studying the dynamical process of the game play. Also, there are usually multiple Nash equilibria in network games, so incorporating behavioural features may provide insights into equilibrium selection.

In this paper, I use (a modified version of) the experience weighted attraction (EWA) model by \cite{camerer1999experience} to formulate agents' behavioural rules. Each agent at each time has a latent variable called attraction level for each possible action and the frequency of choosing a particular action is positively correlated with the attraction of that action. The attractions evolve according to the payoffs agents get as they play the game repeatedly. There are three key parameters that quantify different aspects of agents' behavioural features in the paper:\footnote{Note that there is one additional parameter $\kappa_i$ associated with one additional state variable called experience $N_i(t)$ in \cite{camerer1999experience}, which I do not consider in this paper for simplicity. The updating rule of $N_i(t)$ in \cite{camerer1999experience} is independent of that of attractions. It can be shown that the behavioural rule in this paper is equivalent to the case where $N_i(t)$ has reached its steady state.} $\lambda_i$ -- the sensitivity of action to attraction level (decision accuracy), $\psi_i$ -- the depreciation rate of attractions (e.g., due to forgetfulness), and $\eta_i$ -- the weight on unselected choices when updating attraction (e.g., due to limited attention).\footnote{As noted in \cite{camerer1999experience}, when $\eta_i=1$, the model is reduced to belief-based learning, and when $\eta_i=0$, it represents reinforcement learning.} Thus, the mapping from attractions to actions and the dynamic adjustment rule of the attractions prescribe an agent's behavioural rule.\footnote{This framework is flexible as other behavioural features could easily be included in the setting. For example, imitating-the-best behavioural rule could be added by letting the attraction of an action depend on neighbours' payoffs from that action.}

The learning model defines a dynamical system in which the state at any given time is all agents' attraction vectors. The questions of interest are how the long-term behaviour of the dynamical system (and thus, equivalently, the long-term action profile of the population) will be under the EWA behavioural rule and how the properties of the long-term action profile may depend on the parameters in the model (i.e., bifurcation analysis in the language of dynamical systems).

For concreteness, I fix the stage game to be $2 \times 2$ network coordination games (see e.g., Chapter 12 of \cite{goyal2023} for a review). For interactions between two agents, coordinating on the same action generates higher payoffs than mis-coordination does. In the network environment, each individual plays the game against her neighbours and the payoff she gets is the summation of the payoffs from the interactions with each neighbour. The Nash equilibrium (NE) of the game is described by the concept of network cohesiveness (\cite{morris2000contagion}), which means that in a NE, for each agent, the fraction of neighbours choosing the same action as the agent in question should be above some threshold, with the threshold depending on which particular action is considered and on the payoff matrix.

The number of Nash equilibria of the above network coordination game can be large for general networks. This paper studies the game from a dynamical perspective by analysing the properties of the long-term action profile with bounded-rational agents who conduct the EWA learning. More specifically, some questions explored include whether players will play any NE in the long run, whether the possible long-term action profile can be richer than NE, which NE is likely to be played and how it depends on the model primitives, etc.

I first show that the long-term action pattern will generally converge to some action profile and cannot exhibit more complicated dynamics, such as limit cycle or chaos due to the dynamical game being a cooperative dynamical system in the sense of \cite{hirsch1985systems}.

I then show that when agents are very inaccurate (low $\pmb{\lambda}$) or forgetful (high $\pmb{\psi}$), then the long-term action pattern must be that all agents favour the risk-dominant option rather than the efficient option (when there is a conflict between the two), independent of their initial attractions (and other model parameters). In other words, even if all agents favour the efficient option initially, they will prefer the risk-dominant choice to the efficient one ultimately. The reason behind this is that high noises in action and high depreciation of past attractions will attenuate agents' initial attractions and due to the risk of selecting the efficient choice, people will not play it when the overall tendency in the population to play the efficient action is not strong, which will ultimately drive the whole population to the risk-dominant option.

I next study the opposite case where agents have long memory or high accuracy. In this case the set of possible long-term action profiles depends on the parameter $\pmb{\eta}$ -- the weight on unselected action when updating attractions -- and on the payoff matrix. If the payoff from mis-coordination is not very negative, then the long-term action pattern of the dynamical game can be ``close" to any Nash equilibrium starting with some initial conditions. The result is the opposite when mis-coordination generates very negative results. In particular, some Nash equilibrium may not be supported by the behavioural model.

The intuition behind these is that when an agent pays limited attention to unselected actions, his attractions tend to only depend on how satisfied she is from playing the action she chose without comparing it to the payoff she could have gained had she chose the other strategy. When the result of mis-coordination is not very negative, agents could still get a ``satisfactory" payoff when they did not make a best response, but if there is a large payoff loss in the case of mis-coordination, then they might not be satisfied with the results of some partial coordination profile, which induces a stricter requirement as compared to NE for an action profile to be stable. 

I then explore the monotonic and non-monotonic relationships between the set of possible long-term action profiles and parameters. Following similar reasoning as above, when people have high decision accuracy or low forgetfulness, the number of pure-strategy profiles that could be played in the long-run in the behavioural model has a monotonic relationship with $\pmb{\eta}$ in the sense that if a pure-strategy profile can be played under $\pmb{\eta}$, then it can also be played after either a monotonic increase or decrease of $\pmb{\eta}$, with the direction of the monotonic change depending on whether the mis-coordination payoffs compared to the coordination payoffs are sufficiently negative or not. The number of possible long-term action profiles in general has a non-monotonic relationship with $\pmb{\psi}$ and $\pmb{\lambda}$. In particular, it can be larger than the number of Nash equilibrium when people conduct belief-based learning.

The above analyses focus on all possible long-term action profiles. A natural question is then which of them (if there are multiple one) people tend to play in the long term given their initial attractions. This brings to the study of influence questions in terms of whose initial attractions are important to the actions of the whole population. I show that when all agents' initial attractions are sufficiently close to an unstable fixed point, the choice which the population will play in the long run is determined by a weighted average of people's initial attractions, with the weights, which reflect agents' influences, given by the principal left eigenvector associated with a Jacobian matrix. The influence of an agent is decreasing with her forgetfulness and increasing with her decision accuracy under some conditions.

The use of the above eigenvector to represent agents' influences, in principle, is only valid when agents initial attractions are arbitrarily close to the neutral state. The numerical calculations confirm that its prediction accuracy is decreasing with the variance of people's initial attractions, but it is still relatively high when the initial conditions are not very close to the neutral state. I then briefly analyse numerically how the distributions of behavioural parameters across agents in different network locations may have an impact on the long-run action profile. The results show that the long-term action profile tends to be consistent with the predispositions of those who have a high decision accuracy and network centrality. Moreover, the influence of a highly accurate agent is larger when the average accuracy level is small as compared to when it is high.

\section{Literature Review}
\label{sec:liter}
I now relate my work to the literature. First, my research builds on the model of coordination games on networks.\footnote{see e.g., Chapter 12 of \cite{goyal2023} for a review.} It is well known that the set of pure-strategy Nash equilibria in this game can be characterised by the concept of network cohesiveness (e.g., \citet{morris2000contagion} and \citet{goyal2023}): an action profile is a Nash equilibrium if and only if the set of agents choosing the same action forms a $r$-cohesive set, with the parameter $r$ depending on the specific action and on the payoff matrix. Thus, the number of pure-strategy Nash equilibria of the game is at least two (all people coordinating on either option) and can be large for arbitrary network structures. This brings to the questions on the convergence of action dynamics and on the equilibrium selection of the game. 

A strand of this literature studies game play under persistent randomness in individual choices (e.g., \citet{blume1993statistical},  \citet{ellison1993learning}, \citet{ellison2000basins}, \citet{jackson2002formation}, \citet*{kandori1993learning}, \cite{young1993evolution}, \citet{young2001competition}, etc.). A setting frequently analysed is that when agents are supposed to make a decision, they make a myopic best response with a high probability and make a mistake with a low probability. The agents in these models are bounded rational in the sense that (1) their actions are myopic in each period and that (2) they might make errors when making decision. These papers study the long-run likelihood distribution of action profiles. A concept arises from this setting is called stochastic stability (\citet{foster1990stochastic}), which describes the states that are likely to appear if the probability of agents making mistakes becomes arbitrarily small.\footnote{More formally, stochastically stable states are those states that have strictly positive limiting probability if the probability of making an error goes to zero (see \citet{foster1990stochastic}).}

The stochastically stable states may depend on the network structure, the behavioural rules, and how people make mistakes. For example, if people randomly select an option (independent of the payoff structure) when making a mistake, it has been shown that the unique stochastically stable state is that all players coordinate on the risk-dominant option if the interaction structure is a complete network (\citet{kandori1993learning}) or a circle network (\cite{ellison1993learning}). As noted in \cite{jackson2002formation}, coordination on either the risk-dominant or the efficient option is stochastically stable if the interaction structure is a star network.

When players conduct log-linear best responses, \citet{blume1993statistical} shows that all agents playing the risk-dominant choice are the unique stochastically stable state when players are located on a lattice. This outcome can be generalised to any network structure as the stochastically stable state under log-linear best responses is the action profile that maximises the potential for potential games (e.g., \citet{blume1993statistical}, \citet{young2001competition} and Chapter 6, \citet{young2020individual}), which is the state where all agents choose the risk-dominant option in the networked coordination game.

The long-run action patterns also depend on which behavioural rules agents apply. For example,  \citet{robson1996efficient} show that efficient equilibrium rather than risk-dominant equilibrium is the unique stochastically stable state when players are randomly matched to play a symmetric $2\times 2$ coordination game in each period by following the imitating-the-best rule. Equilibrium selection favouring the efficient option is also observed in \cite{alos2008contagion} who show that under some (mild) assumptions on the network structure, only the efficient option can spread to the whole population from an initially small set of adopters if people follow the imitating-the-best behavioural rule. Note that these results are in contrast to most of the stochastic stability literature above (e.g., \cite{ellison1993learning}) that favours risk-dominant option as well as the contagion outcome from best response dynamics in \citet{morris2000contagion} which rules out the spread of the efficient option from an initially finite adopters.

The analysis in my paper is distinct from and more general than the above literature. Firstly, stochastic stability, by definition, only studies long-term behaviour that does not depend on initial conditions which require vanishing magnitude of noises. In contrast to this, my research allows non-vanishing noises and whether initial condition plays a role is endogenously determined by the parameters in the model. Secondly, my work contains more aspects of agents' behavioural features (memory/forgetfulness and reinforcement learning), in addition to decision errors. The above literature has provided evidence that agents' behavioural features can impact equilibrium selection, but they typically only consider one single decision rule in each paper. Analysing all those features together avoids the ad hoc issue of the above papers, at least to some extent. Indeed, in certain ranges of the parameters such as when the decision error is large or when people are forgetful, the unique behavioural equilibrium is favouring the risk-dominant choice, which is consistent with many of the above papers (e.g., \citet{blume1993statistical}). However, other configurations of long-term action profile can occur if the behavioural features vary. Thus, in contrast to most other papers that focus on one given decision rule, this paper provides a characterisation on how the qualitative properties of long-run action patterns may change with multiple behavioural features of agents.

This paper is also related to the broad literature that studies convergence properties of actions under adaptive learning\footnote{For a review of models of adaptive learning, see, e.g., Chapter 15 of \citet{dhami2016foundations}.} (e.g., \citet{robinson1951iterative}, \citet{miyasawa1961convergence}, \citet{monderer1996potential}, \citet{shapley1964some}, and \citet{fudenberg1993learning}) and evolutionary games on networks (e.g., \citet*{debarre2014social}, \citet*{hauert2004spatial}, \citet*{nowak2010evolutionary}, \citet*{roca2009evolutionary}, \citet{zhou2021aspiration}, and \citet{zukewich2013consolidating}).\footnote{For a review of graphical evolutionary games, see, for instance, \citet*{szabo2007evolutionary}.} The literature studying convergence properties of behavioural learning processes typically does not consider networks or multiple behavioural features and its main focus is on whether the action dynamics converge instead of more general equilibrium selection questions, which are studied in this paper. The research on graphical evolutionary games is usually based on numerical methods (\citet{correia2022asymmetric}), or based on special assumptions such as weak selection limit (e.g., \citet{allen2017evolutionary} and \citet{zhou2021aspiration}). The analyses of how the behavioural features captured in the EWA model impact the outcomes of dynamical network games are novel compared to them.

I now discuss the related literature in terms of the specific methodology used and the results obtained in the research. The experience-weighted attraction (EWA) model (\citet{camerer1999experience}) I use in the paper has the advantage that it can include multiple aspects of behavioural features, such as reinforcement learning, forgetfulness, and noises. It can also be naturally generalised to include other behavioural aspects, such as aspiration level and reinforcing the best, which are also briefly discussed in the paper. The empirical relevance of the EWA model has been shown in the studies (e.g., \citet{camerer1999experience}, \citet*{camerer2004behavioural}, \citet{camerer2015behavioral}, and see also the review by \citet{dhami2016foundations}). The EWA model can predict people's actions in experiments better than other models such as belief-based, quantal response, and reinforcement learning in most cases (\citet{camerer2015behavioral} and \citet{dhami2016foundations}).

A few papers have characterised the properties of action patterns under the EWA learning in certain types of games. For example, early work (\citet*{sato2002chaos} and \citet*{sato2003coupled}) show that reinforcement learning (a special case of the EWA model) can result in chaotic behaviour in even simple low dimensional games. \citet*{galla2013complex} and \citet*{sanders2018prevalence} show that in random complicated games, the action patterns under EWA learning have several possible outcomes, ranging from unique fixed states, to a large number of fixed states, and to chaotic behaviour, depending on the correlation of payoffs across players and the level of forgetfulness. \citet{pangallo2017taxonomy} study the long-term properties of action patterns in generic $2 \times 2$ games. They characterise how the payoff matrix and the irrationality level of players reflected in the EWA parameters determine whether the dynamics exhibit convergence to (pure-strategy or mixed) fixed point(s) or limit cycle and chaotic behaviours. The relationship between action patterns and $\pmb{\psi}$ and $\pmb{\lambda}$ in my paper is consistent with the results in \citet{pangallo2017taxonomy}. My research is distinct from theirs in four aspects: the consideration of network structures, the inclusion of reinforcement learning as one behavioural feature, the relaxation of the homogeneity assumption on behavioural rules across agents, and the study of the influence questions.

On the analysis of agents' influences, the importance of some eigenvectors also appears in other contexts, such as the DeGroot learning model (\cite{degroot1974reaching})\footnote{See also \citet*{demarzo2003persuasion}, \cite{golub2010naive}, and \citet{golub2012homophily} for related studies and, e.g., Chapter 13 of \citet{goyal2023} and \cite{golub2017learning} for a survey of network learning models.} and diffusion models on networks (e.g., \cite{newman2018networks}).\footnote{Eigenvectors are also important to network interventions problems (see e.g., \cite*{galeotti2020targeting}, \citet{galeotti2021discord}, etc.).} For example, in the DeGroot learning model, the importance of agents is summarised by the eigenvector centrality of the belief updating matrix. In this paper the Jacobian matrix evaluated at a neutral point plays a similar role as the updating matrix does in the learning model in that both describe how the states (attractions or beliefs) of agents influence each other. The influence vector in this work is distinct from that in the learning model in that it incorporates the information about agents' behavioural features into the network structures, which provides a compact representation showing how behavioural features interact with network structures in impacting the action pattern of the population.\footnote{Another significant difference as will be described in length in Section \ref{sec:influ} is that the DeGroot learning model is a linear system while the model in this paper is nonlinear so the prediction made by the left eigenvector only works when people's initial attractions are close to the unstable fixed point. However, the numerical calculations show that it can work reasonably well under a relatively large range.} In the network SI model, the eigenvector centrality associated with the adjacency matrix describes the probabilities of different agents getting infection in early periods (\cite{newman2018networks}), which serves an analogous role as the left eigenvector did in this research. One difference between the two is that the eigenvector centrality in my paper not only impacts the outcomes in early periods but also has implications for what choice will be selected by the population in the long run.

To conclude, my research is distinct from other literature mainly in that I smoothly consider multiple aspects of behavioural features and heterogeneities of them across agents in a single study. Without imposing an ex ante reduction of the model parameters, I explore what patterns there are in the high-dimensional (so rather uninterpretable) mapping from the model primitives to the long-term action profile. The bifurcation analysis of the dynamical game and the study of the agents' influences show how the elements in behavioural economics might produce insights into the study of economics of networks. The analysis is also more flexible than previous research and can be extended to include other behavioural rules and study other games.\footnote{For example, prisoner's dilemma and anti-coordination games.}

The remaining part of the paper is organised as follows: Section \ref{sec: set} sets up the basic framework. Section \ref{sec:beha} considers how the set of possible long-term action patterns depends on the model primitives, and Section \ref{sec:influ} studies the influence questions in terms of whose initial predispositions are important to the action profile of the population. Section \ref{sec: extension} briefly discusses some model extensions. Section \ref{sec:con} concludes.

\section{The Model}
\label{sec: set}
The paper studies network games with bounded rational agents. There are three ingredients in the model: stage game, network structure, and behavioural features.
\subsection{Stage Game} 
Consider a $2 \times 2$ coordination game as shown in Table \ref{table-basic}:
\begin{table}[h!]
\centering
\caption{a coordination game}
\label{table-basic}
 \begin{tabular}{|c |c | c |} 
 \hline
  & $C$ & $D$ \\ [0.5ex]
 \hline
 $C$ & z, z & y, x \\ [0.5ex]  \hline
 $D$ & x, y & w, w \\
  [1ex] 
 \hline
 \end{tabular}
\end{table}

\begin{assumption}
\label{ass: payoff-basic}
(payoff structure): $z>x$, $w>y$, $w>x$, $z>y$, $w>0$ and $z>0$.
\end{assumption}

Thus, the game has two pure-strategy Nash equilibria: (1) all players choosing action $C$ (denoted by $s^*_1=s^*_2=0$) and (2) all players choosing action $D$ ($s^*_1=s^*_2=1$). That is, coordinating on either choice is a Nash equilibrium and mis-coordination is not. Moreover, mis-coordination results in payoffs that are worse than either coordination outcome for both players. The assumption of $w>0$ and $z>0$ indicates that coordination results in positive payoffs. This assumption is made for the reinforcement learning described later and is without loss of generality as shown in Appendix.
\subsection{Networks} It is assumed in the paper that agents are located on an undirected network $G=(N, E)$, where $N$ is the set of individuals, $N=\{1, 2, ..., n\}$, and $E$ is the set of edges. Denote the adjacency matrix associated with the graph also by $G$. Two agents $i$ and $j$ are connected if and only if $G_{ij}=G_{ji}=1$. For each individual $i \in N$, the neighbours of $i$ is the set of agents with whom $i$ is connected, and this set is denoted by $N_i$. Without loss of generality, assume that the network $G$ is connected (if not, we can analyse each component separately).

\begin{assumption} (network):
\label{ass:network}
The undirected network $G$ is connected. 
\end{assumption}

Each individual $i$'s payoff in each period is the sum of the payoffs from the interaction with each of her neighbours.
\begin{align}
\label{net-utility}
u_i(\mathbf{s}) = \sum_{j \in N_i} u_i(s_i, s_j)
\end{align}

The actions studied in this paper are generally mixed strategies as will be discussed later. Denote the probability/frequency of player $j$ playing strategy $D$ at time $t$ by $p_{j}(t)$ (or simply $p_j$ if omitting the time notation), then the probability/frequency of player $j$ playing strategy $C$ at time $t$ is $1-p_{j}(t)$. Denote $\mathbf{p} = [p_1, p_2, ..., p_n]$ as the strategy profile of the population and $p_{-i}$ as the strategy of all individuals other than $i$.

The payoff a player $i$ obtains is then:
\begin{align}
\label{expected utility}
u_i(\mathbf{p})=p_i u(p_i=1, p_{-i}) + (1-p_i) u(p_i=0, p_{-i}) 
\end{align}

where
\begin{align}
\label{expected utility2}
\begin{split}
 & u(p_i=1, p_{-i}) = \sum_{j \in N_i} \left( p_j w + (1-p_j)x \right)\\
&  u(p_i=0, p_{-i}) = \sum_{j \in N_i} \left(p_j y + (1-p_j)z \right)   
\end{split}
\end{align}

\subsection{Behavioural Rules} 
This paper uses (an adjusted version) of the EWA model by (\cite{camerer1999experience}) as agents' behavioural model. Each individual $i$ has an attraction vector representing the attraction of each strategy at each time: $a_{i}(t) = [a_{i,0}(t),a_{i,1}(t)]$, $\forall i, t$. For each individual, the attraction vector determines action in each period according to the following logit response. Define $q_i(t):=a_{i,1}(t)-a_{i,0}(t)$, the difference between the attraction of action $D$ and that of action $C$, then

\par\noindent\textit{attraction $\to$ action:}
\begin{align}
\begin{split}
   p_{i}( t) := & Pr(s_i(t)=1)=\frac{\exp(\lambda_i a_{i,1}(t))}{\exp(\lambda_i a_{i,0}(t))+\exp(\lambda_i a_{i,1}(t))} \\
= & \frac{1}{1+\exp(-\lambda_i q_i(t) )}  
\end{split}
\label{att to str}
\end{align}
where $\lambda_i >0$ reflects the sensitivity of actions to attractions. When $\lambda_i=0$, agents randomly choose actions regardless of the attractions, while if $\lambda_i =\infty$, then the agents will choose the action that has the highest attraction level with certainty. We can see that only the difference in the attraction of the two actions matters, so $q_i(t)$ is sufficient to summarise the information about an agent $i$'s attractions at time $t$.

As in \citet*{sato2002chaos}, \citet{pangallo2017taxonomy}, and \citet{galla2013complex}, I assume that agents play each strategy with some frequency instead of playing one strategy with a specific probability. This assumption transforms the stochastic process of game play into a deterministic one. The main reason for this setting is mathematical tractability, but there are also other justifications. Firstly, in some situations the ``frequency" might be interpreted as the extent of an agent's behaviour (e.g., the extent of altruism) which might evolve gradually. Secondly, when people do not have a strong preference for a particular strategy, they may indeed try different actions with positive frequencies and those frequencies might also evolve gradually. Thirdly, \cite{galla2013complex} have shown by simulations that the deterministic process can approximate the original stochastic process well in most cases. 

Another important ingredient of the behavioural rule is how agents update attractions according to the outcomes of game play, which is defined as follows.

\par\noindent\textit{updates of attractions:}\footnote{This can be derived from the limit of the following discrete case:
\begin{align}
 a_{i,k}(t+ \triangle)=e^{-\psi_i \triangle} a_{i,k}(t) + \left[Pr(s_i(t)=k)+\eta_i (1-Pr(s_i(t)=k)) \right]u(s_i(t)=k, p_{-i}(t) ) \triangle
\end{align} where $k =0, 1$.
Taking the limit $\triangle \to 0$ gives:
\begin{align}
\begin{split}
\dot{a}_{i,k}(t) & =\lim_{\triangle \to 0} \frac{e^{-\psi_i \triangle}-1}{\triangle} a_{i,k}(t) + \left[Pr(s_i(t)=k)+\eta_i (1-Pr(s_i(t)=k)) \right]u(s_i(t)=k, p_{-i}(t) )  \\
& = -\psi_i a_{i,k}(t)  + \left[Pr(s_i(t)=k)+\eta_i (1-Pr(s_i(t)=k)) \right]u(s_i(t)=k, p_{-i}(t) )
\end{split}
\end{align}
}\begin{align}
\begin{split}
\label{attraction updates}
&\dot{a}_{i,1}(t) =-\psi_i a_{i,1}(t) + (p_{i}(t)+\eta_i(1-p_{i}(t))) \sum_{j \in N_i} \left(p_{j}(t) w + (1-p_{j}(t))x \right) \\
& \dot{a}_{i,0}(t) =-\psi_i a_{i,0}(t) + (1-p_{i}(t)+\eta_i p_{i}(t)) \sum_{j \in N_i}   \left(p_{j}(t) y + (1-p_{j}(t))z \right)
\end{split}
\end{align}
where the notation $\dot{a}$ represents its time derivative $\frac{d a}{d t}$. Thus, agents' attractions follow an exponential depreciation rate $\psi_i >0$. The depreciation might be due to forgetfulness or due to people intentional paying more attention to new information and less to past information.

$\eta_i \in [0,1]$ is the extent to which forgone payoffs are considered in attraction updates. Consider the two extreme cases. When $\eta_i=1$, as shown in \citet{camerer1999experience}, the model reduces to belief-based learning since the attraction of each action will be fully updated regardless of whether the agent plays that action.\footnote{To see this, note that when $\eta_i=1$
\begin{align}
\begin{split}
 \label{belief-based}
&\dot{a}_{i,1}(t) =-\psi_i  a_{i,1}(t)  + \sum_{j \in N_i} \left( p_{j}(t)  w + (1-p_{j}(t) )x \right)  \\
& \dot{a}_{i,0}(t) =-\psi_i a_{i,0}(t)  + \sum_{j \in N_i} \left(p_{j}(t)  y + (1-p_{j}(t) )z \right)   
\end{split}
\end{align}} When $\eta_i=0$, the updates of attractions reflect reinforcement learning as the attraction update of a particular action by an agent is proportional to the frequency with which that action is played by the agent.\footnote{More precisely, when $\eta_i=0$
\begin{align}
\begin{split}
 \label{reinforce}
&\dot{a}_{i,1}(t) =-\psi_i  a_{i,1}(t)  + p_i(t) \sum_{j \in N_i} \left( p_{j}(t)  w + (1-p_{j}(t) )x \right)  \\
& \dot{a}_{i,0}(t) =-\psi_i a_{i,0}(t)  + (1-p_i(t)) \sum_{j \in N_i} \left(p_{j}(t)  y + (1-p_{j}(t) )z \right)   
\end{split}
\end{align}}

The baseline model assumes that agents' behavioural parameters are time-invariant and are within the following range.

\begin{assumption} (behavioural parameters):
\label{ass:finite}
$\psi_i > 0$, $\lambda_i > 0$ and $\eta_i \in [0, 1]$, for any $i$.
\end{assumption}

Recall that $q_{i}(t) :=a_{i,1}(t)-a_{i,0}(t)$ summarises an agent's individual state at each time $t$. For brevity ignore the time $t$ notation and define $\mathbf{q}=[q_1, q_2, ... q_n]^T$ and $\dot{\mathbf{q}}=[\dot{q_1}, \dot{q_2},... \dot{q_n}]^T$. It can be easily seen that $\dot{\mathbf{q}}$ only depends on $\mathbf{q}$ and is given by, for each $i$,
\begin{align}
\begin{split}
\label{big q}
 \dot{q}_{i} :=  F_i(\mathbf{q})= & -\psi_i q_{i} + (p_{i}+\eta_i(1-p_{i})) \sum_{j \in N_i} \left(p_{j} w + (1-p_{j})x \right) \\ &
  - (1-p_{i}+\eta_i p_{i}) \sum_{j \in N_i}   \left(p_{j} y + (1-p_{j})z \right)
  \end{split}
\end{align}
where $p_i$ and $p_j$ are functions of $\mathbf{q}$ and are defined in (\ref{att to str}). Then (\ref{big q}) defines a deterministic autonomous dynamic system $\dot{\mathbf{q}}=F(\mathbf{q})$, where the attraction difference vector $\mathbf{q}$ completely characterises the state of the system at any time. The dynamical system completely characterises the dynamics of agents' attractions and thus action patterns.

\subsection{Parameter Space and Notations} In this paper, a scalar is denoted by lowercase letters (e.g., payoff $w$). A vector is represented with a bold mode (e.g., attraction difference $\mathbf{q}$). A particular element of a vector is represented with a subscript without the bold mode (e.g., agent $i$'s attraction difference $q_i$). A matrix is represented by a capital letter (e.g., the adjacency matrix $G$).

For comparison of vectors, the notation is that $\mathbf{a} \geq \mathbf{b}$ represents that $a_i \geq b_i$ for any $i$, and $\mathbf{a} > \mathbf{b}$ indicates that $a_i > b_i$ for any $i$.

The parameters in the model include the payoff matrix $\Pi$, the network $G$, and three behavioural parameters $\pmb{\psi}$, $\pmb{\lambda}$ and $\pmb{\eta}$, where $\pmb{\psi} = \{\psi_i \} _{i=1,2,...,n}$, $\pmb{\lambda}=\{ \lambda_i \}_{i=1,2,...,n}$, and $\pmb{\eta} =\{ \eta_i \} _{i=1,2,...,n}$. Denote the set of parameters as $\Gamma=\{\Pi, G, \pmb{\psi}, \pmb{\lambda}, \pmb{\eta} \}$.\footnote{With these notations in hand, the autonomous dynamical system (\ref{big q}) can be neatly represented in the following matrix form
\begin{align}
\begin{split}
\label{big q matrix}
 \dot{\mathbf{q}} = &-\pmb{\psi}  \circ \mathbf{q} + (\mathbf{p}+ \pmb{\eta} \circ (\pmb{1}-\mathbf{p})) \circ \left( (w-x) G \mathbf{p} + x G \pmb{1} \right) \\ 
 & - (\mathbf{1}-\mathbf{p}+\pmb{\eta} \circ \mathbf{p}) \circ \left ( (y-z) G \mathbf{p} + z G \pmb{1}\right)
 \end{split}
\end{align}
where $\circ$ represents element-wise matrix multiplication.}

Given the initial attraction difference $\mathbf{q}(0) \in \mathcal{Q}$, the evolution of $\mathbf{q}$ is deterministic. Also, recall that there is a one-to-one relationship between $\mathbf{q}$ and $\mathbf{p}$ as $p_i=\frac{1}{1+e^{-\lambda_i q_i}}$ for any $i$, so $\mathbf{q}$ fully summarises the action profile of the population.\footnote{Indeed, the evolution rule of $\mathbf{p}$ is
\begin{align}
\begin{split}
\label{big p}
 \dot{p}_{i}= &  p_i (1-p_i) \left[\psi_i \ln \left( \frac{1}{p_i}-1 \right) + \lambda_i \tau_1(p_i) u(p_i=1, p_{-i})- \lambda_i \tau_0(p_i) u(p_i=0, p_{-i})  \right] 
  \end{split}
\end{align}
for any $i$, where $\tau_1(p_i)=p_{i}+\eta_i(1-p_{i})$ and $\tau_0(p_i)=1-p_{i}+\eta_i p_{i}$. In the paper, I will use the dynamical system of $\mathbf{q}$ in (\ref{big q}) and that of $\mathbf{p}$ in (\ref{big p}) interchangeably depending on which one is more convenient to the analysis of a particular question.
}

For reference, recall that $\mathbf{q^*}$ is a fixed point of the dynamical system $\dot{\mathbf{q}}= F(\mathbf{q})$ if $F( \mathbf{q^*}) =0 $. A fixed point is stable if (1) there exists a neighbour $N^\delta(\mathbf{q^*})$ such that whenever $\mathbf{q}(0) \in N^\delta(\mathbf{q^*})$, $\lim_{t \to \infty} \mathbf{q}(t) = \mathbf{q^*}$, and (2) for any $N^\epsilon(\mathbf{q^*})$, there exists a neighbour $N^\delta(\mathbf{q^*})$ such that whenever $\mathbf{q}(0) \in N^\delta(\mathbf{q^*})$, we have that $\mathbf{q}(t) \in N^\epsilon(\mathbf{q^*})$ for any $t \geq 0$. These definitions are well known in dynamical systems and can be found in e.g., \cite{strogatz2018nonlinear}.

The central task in the paper is to study how the properties of the long-term behaviour of the dynamical system -- equivalently, the properties of the long-term action pattern of the game -- depend on the model primitives $\Gamma=\{ \Pi, G, \pmb{\psi}, \pmb{\lambda}, \pmb{\eta} \}$. The common questions studied in dynamical systems apply, which include, for example, whether the action profile converges, how many steady states there are, how the long-term action profile depends on initial conditions, etc.

\section{Behavioural Equilibrium}
\label{sec:beha}
This section discusses how the set of possible long-term outcomes of the game play depends on the model primitives. I first define the concept of behavioural equilibrium (BE) considered in the paper.

\begin{definition}
\label{def:be}
A mixed-strategy profile $\mathbf{p}^*$ is a behavioural equilibrium (BE) if and only if its corresponding $\mathbf{q}^*$ is a stable fixed point of the dynamical system (\ref{big q}). Denote the set of BE by $\mathcal{B}$.
\end{definition}

According to the definition of stable fixed points (as described in Section \ref{sec: set}), a BE is an action profile under which all agents' attractions of actions and thus frequencies of actions remain constant over time. In addition, there is some nontrivial set of initial conditions from which the action pattern converges to that action profile. In other words, loosely speaking, a BE is a stationary action profile that can be seen in the long-run.

It is well known that behaviours in a high dimensional nonlinear dynamical system may have attractors that are not limited to fixed points, such as limit cycles and even chaos, but the next proposition summarises that the action profile of the dynamical game in this paper must generally converge to some fixed action profile.

\begin{proposition}
\label{pro1}
Suppose Assumptions \ref{ass: payoff-basic}-\ref{ass:finite} hold. Then the action profile $\mathbf{p}(t)$ converges to a strategy profile from almost all initial attraction differences $\mathbf{q}(0)$.
\end{proposition}

\begin{proof}
It is equivalent to study the dynamical system of $\mathbf{q}$ as in (\ref{big q}), which has
\begin{align}
\label{coop}
 \frac{\partial F_i(\mathbf{q})}{\partial q_j}  = G_{ij} \frac{d p_j}{d q_j} \left[(p_i + \eta_i (1-p_i))  (w-x) - (1-p_i + \eta_i p_i) (y-z) \right] \geq 0
\end{align}
for any $i \neq j$, because $G_{ij} \geq 0$, $\frac{d p_j}{d q_j} \geq 0$, $p_i + \eta_i (1-p_i) \geq 0$, $1-p_i + \eta_i p_i \geq 0$ , $w>x$, and $y<z$.

Thus, the dynamical system $\dot{\mathbf{q}}=F(\mathbf{q})$ is a cooperative dynamical system as in, e.g., \cite{hirsch1985systems}. Also, $F$ is irreducible since the network $G$ is connected. Moreover, any forward trajectory has compact closure as the trajectory is bounded due to $\pmb{\psi}>\mathbf{0}$ and finite payoffs. 

Then applying Theorem 4.1 in \citet{hirsch1985systems}, the dynamical system $\dot{\mathbf{q}}=F(\mathbf{q})$ converges to some fixed point from almost all initial conditions.
\end{proof}

This statement indicates that the action dynamics will generally converge and cannot exhibit more complicated dynamics, such as limit cycle or chaos, which may occur in other types of games where the payoff matrix is asymmetric (see e.g., \citet{galla2013complex}).

Proposition \ref{pro1} justifies restricting the study of long-term outcomes to behavioural equilibria as in Definition \ref{def:be} which only consider fixed points. A question then to be explored is how the behavioural equilibria relate to the pure-strategy Nash equilibria (NE) of the game. Note that NE is characterised by the concept of $r$-cohesiveness as in \cite{morris2000contagion}.\footnote{Formally, let $r=\frac{z-x}{w-x+z-y}$ and let $N_C^*$ and $N_D^*=N \setminus N_C^*$ be the group of agents choosing action $C$ and $D$ respectively. Then this pure strategy profile is a Nash equilibrium if and only if the group $N_C^*$ is $(1-r)$-cohesive and the group $N_D^*$ is $r$-cohesive (\cite{morris2000contagion}).} For concreteness, I analyse the BE in the context where there is a conflict between risk dominance and efficiency. The consideration of the trade-offs between these two is common in the literature as discussed in Section \ref{sec:liter}.

\begin{assumption}
\label{ass: risk-effi}
Suppose $z>w>x>y$, $w+x>z+y$, $w>0$ and $z>0$.
\end{assumption}

This assumption indicates that, as defined in \cite{harsanyi1988general}, $D$ is the risk-dominant choice and $C$ is the efficient choice.

As the dynamical system is nonlinear, it is generally not possible to analytically characterise the whole set of behavioural equilibria, $\mathcal{B}$. The analyses thus focus on qualitative properties of $\mathcal{B}$ and how they depend on the model parameters $\Gamma$. Recall that the primitives of the model, $\Gamma$, include the network $G$, the payoff matrix $\Pi$, and the three behavioural parameters: attraction depreciation rate $\pmb{\psi}$, decision accuracy $\pmb{\lambda}$, and weight on forgone payoffs $\pmb{\eta}$. We can regard the set of behavioural equilibria, $\mathcal{B}$, as a function of those parameters: $\mathcal{B}(\Gamma)$. The analyses focus on what patterns can be obtained from the high dimensional mapping $\mathcal{B}(\Gamma)$. The first exploration is on some limiting properties (i.e., what properties the set of behavioural equilibria exhibits when some model parameters, such as behavioural features, take extreme values). Proposition \ref{pro2} studies the case where $\pmb{\lambda} \to \pmb{0}$ or $\pmb{\psi} \to \pmb{\infty}$\footnote{Unless otherwise stated, $\pmb{\lambda} \to \mathbf{a}$ or $\pmb{\psi} \to \mathbf{b}$ in the paper means that at least one of the two holds, thus including the case where they both hold i.e., $[ \pmb{\lambda}, \pmb{\psi}] \to [\mathbf{a}, \mathbf{b}]$.} --- agents are inaccurate in decision making or are forgetful, while Proposition \ref{pro: eta} considers the opposite case where $\pmb{\lambda} \to \pmb{\infty}$ or  $\pmb{\psi} \to \pmb{0}$ --- agents are accurate in decision making or have long memory. 

\begin{proposition}
\label{pro2}
Suppose Assumptions \ref{ass: payoff-basic}-\ref{ass: risk-effi} hold. Then for any network $G$ and any $\pmb{\eta}$,  as $\pmb{\lambda} \to \pmb{0}$ or $\pmb{\psi} \to \pmb{\infty}$ (or both), all agents favour the risk-dominant option in any BE (i.e., $q_i^* >0 $, $\forall i$).
\end{proposition}

\begin{proof}
From (\ref{big q}),
\begin{align}
\label{l1-q}
 \dot{q}_{i} > -\psi_i q_{i} - 4 z d_i
\end{align}
Thus, if $q_i < -\frac{4z d_i}{\psi_i}$, then $\dot{q}_{i} >0$ regardless of $q_{-i}$. This means that a fixed point can only occur when $q_i \geq -\frac{4z d_i}{\psi_i}$ for all $i$.

Fixing all model parameters other than $\pmb{\lambda}$, then for any $\mathbf{q}$ such that $q_i \in [ -\frac{4z d_i}{\psi_i} , 0 ]$ for all $i$, we have that as $\pmb{\lambda} \to \pmb{0}$,\begin{align}
\label{p2-1}
    \begin{split}
        \dot{\mathbf{q}}\to -\pmb{\psi} \circ \mathbf{q} + \left(\frac{w+x-y-z}{2}\right) \frac{\pmb{1}+\pmb{\eta}}{2}  \circ G \pmb{1}
    \end{split}
\end{align}
By the assumption of the payoff matrix, $w+x-y-z>0$, so when $\mathbf{q}$ is such that $q_i \geq -\frac{4z d_i}{\psi_i}$ for all $i$, as $\pmb{\lambda} \to \pmb{0}$, it must be that $\dot{q}_j>0$ whenever $q_j \in [ -\frac{4z d_i}{\psi_i} , 0 ]$ since $\dot{q}_j$ is nondecreasing in $q_i$, where $i \neq j$. Thus, $F(\mathbf{q}^*) \neq \pmb{0}$ when $\mathbf{q}^* \ngtr \pmb{0}$. In other words, any $F(\mathbf{q}^*)=\pmb{0}$ must have $\mathbf{q}^*>\pmb{0}$. This completes the proof under $\pmb{\lambda} \to \pmb{0}$.

When $\pmb{\psi} \to \pmb{\infty}$, again fixed points can only occur when $q_i \geq -\frac{4z d_i}{\psi_i}$ for all $i$. When $\pmb{\psi} \to \pmb{\infty}$, $-\frac{4z d_i}{\psi_i} \to 0$. This means that fixing all model parameters other than $\pmb{\psi}$, when $\pmb{\psi} \to \pmb{\infty}$, (\ref{p2-1}) holds for any $\mathbf{q}$ such that $q_i \in [ -\frac{4z d_i}{\psi_i} , 0 ]$, $\forall i$. The statement then follows based on similar argument as in the case of $\pmb{\lambda} \to \pmb{0}$.
\end{proof}

Proposition \ref{pro2} shows that when agents are very inaccurate (low $\pmb{\lambda}$) or when they are very forgetful (high $\pmb{\psi}$), then the long-term action pattern must be that all agents prefer the risk-dominant option to the efficient option, regardless of their initial attractions towards the two choices. Thus, even if all agents are in favour of the efficient option initially, they will favour the risk-dominant option more than the efficient option ultimately. 

The intuition is that high inaccuracy in action (low $\pmb{\lambda}$) and high depreciation of past attractions (high $\pmb{\psi}$) will attenuate agents' attractions to the efficient outcome. Since the efficient option is risky in the setting, agents will not play it when the overall tendency of playing it in the population is not strong enough, which will gradually drive the whole population to favour the risk-dominant option. Note also that in contrast, if all players favour the risk-dominant option, then even if that attraction is not strong, people will continue to favour it precisely due to its risk dominance (when others randomly choose actions, a player gains by choosing the risk-dominant option as compared to selecting the other choice).

The following example illustrates the ideas in a 2-player game.
\\
\par\noindent\begin{example}\textbf{1:}
Figure \ref{vector} shows the vector field of a two-player coordination game for different values of parameters. The payoffs are $\Pi_{00}=z=4$, $\Pi_{01}=y=-2$, $\Pi_{10}=x=1$, and $\Pi_{11}=w=2$. Thus, action $D$ is the risk-dominant option while action $C$ is the efficient option. Suppose $\eta_1=\eta_2=1$. Figure \ref{vector}-(a) and \ref{vector}-(b) show a bifurcation of the system where there are two behavioural equilibria\footnote{There are three fixed points in Figure \ref{vector}-(a), but only the bottom left and the upper right are stable. The middle fixed point is unstable.} in Figure \ref{vector}-(a) when $\pmb{\psi}$ is small and $\pmb{\lambda}$ is large ($\psi_1=\psi_2=0.5$, $\lambda_1=\lambda_2=1$) and there is a unique behavioural equilibrium in Figure \ref{vector}-(b) when $\pmb{\psi}$ is large and $\pmb{\lambda}$ is small ($\psi_1=\psi_2=1$, $\lambda_1=\lambda_2=0.5$).

\begin{figure}[H]
  \centering
  \caption{Vector field of a two-player game}
  \begin{minipage}[b]{0.49\textwidth}
  \subcaption{$\psi_i=0.5$, $\lambda_i=1$}
    \includegraphics[width=\textwidth]{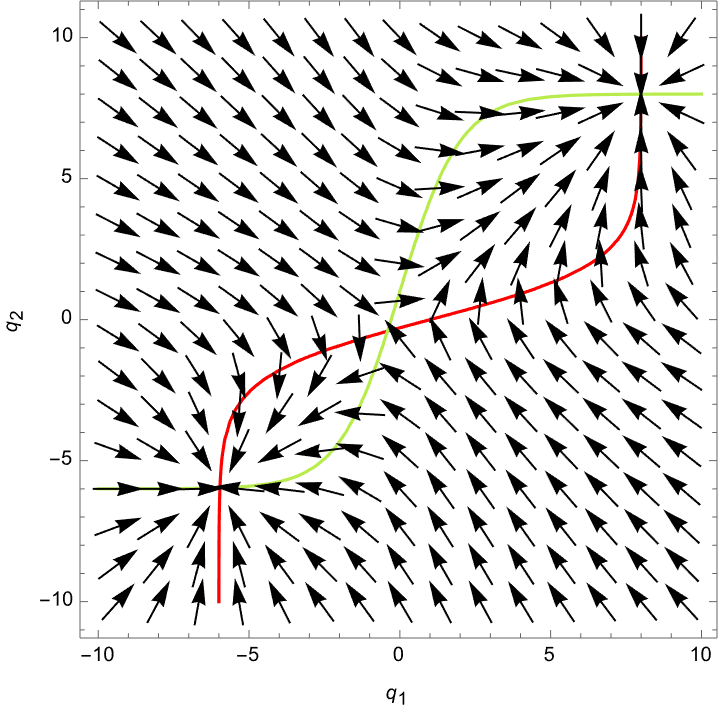}
  \end{minipage}
  \hfill
  \begin{minipage}[b]{0.49\textwidth}
  \subcaption{$\psi_i=1$, $\lambda_i=0.5$}
    \includegraphics[width=\textwidth]{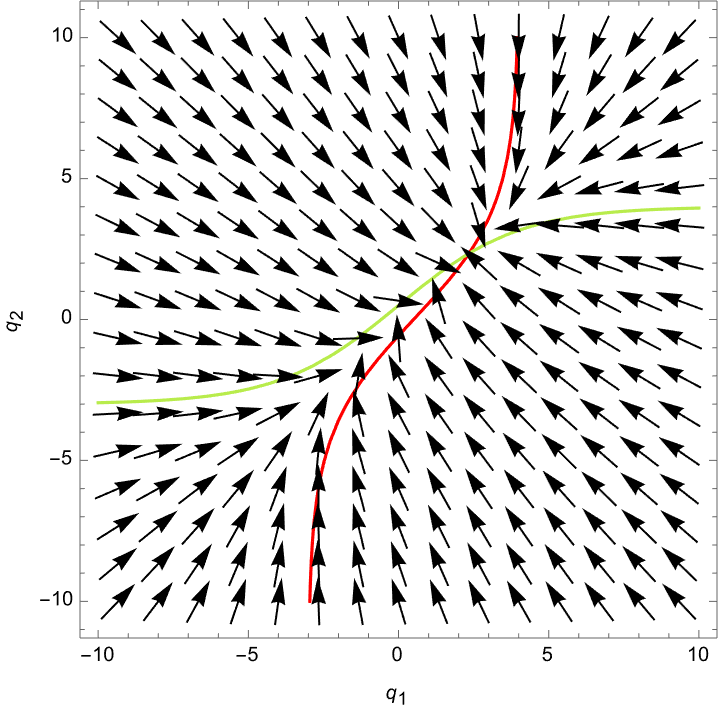}
  \end{minipage}
   \label{vector}
 \caption*{\small Notes: These Figures plot the vector field of the dynamical system \ref{big q} under two different parameter values. The green curve represents the isocline $\dot{q}_2$=0 while the red curve represents the isocline $\dot{q}_1$=0. The intersections of the two curves are fixed points. The black arrows represent the direction of ($\dot{q}_1, \dot{q}_2$). The Figures are plotted in Mathematica.}
\end{figure}

In Figure \ref{vector} -(a), when the two players have relatively high memory and high decision accuracy, they can favour either action in the long run depending on their initial attractions. In Figure \ref{vector} -(b), however, no matter how much agents favour the efficient option initially, that attraction will erode overtime and ultimately be surpassed by the risk-dominant option. Note that this change from two stable fixed points (plus one unstable fixed point) to one stable fixed point as parameters vary is well known as Pitchfork bifurcation in dynamical systems (e.g., \citet{strogatz2018nonlinear}).

\end{example}

Recall that in the stochastic stability literature (e.g., \citet{blume1993statistical}, \citet{young2001competition}, etc.), all players coordinating on the risk-dominant option is the unique stochastically stable state in network coordination games under log-linear best responses. An assumption embedded in the concept of stochastic stability is that the error in decision making is arbitrarily small. In contrast, Proposition \ref{pro2} shows that the selection of the risk-dominant choice can also occur when the error in decision making is large (low $\pmb{\lambda}$). In essence, the selection of the risk-dominant choice under these two opposite conditions lies in different (but intrinsically equivalent) properties of the risk-dominant option. The reason for risk-dominant option being a stochastically stable state is that an agent needs less than (more than) half of neighbours playing the risk-dominant (efficient) option in order for her optimal action to be playing the risk-dominant (efficient) choice. In this paper, risk-dominant option being a unique BE (under low accuracy) lies in that it is optimal for an agent to choose the risk-dominant option if others are randomising actions.

I next consider the opposite case where agents have high memory ($\pmb{\psi} \to \pmb{0}$) or high accuracy ($\pmb{\lambda} \to \pmb{\infty}$). It can be shown that in many scenarios under this limit, the behavioural equilibrium will be arbitrarily close to pure-strategy profile. Intuitively, if $\pmb{\psi} \to \pmb{0}$, then an agent's attraction difference can be arbitrarily large in the long run and therefore she will choose one of the actions with (near) certainty, and if $\pmb{\lambda} \to \pmb{\infty}$, then an agent will choose (almost) pure strategy unless she has the same attraction of the two choices.

A question is then which pure-strategy profiles can be ``approximately" played in the long term when $\pmb{\psi} \to \pmb{0}$ or $\pmb{\lambda} \to \pmb{\infty}$. Especially, do those profiles include Nash equilibrium? How do the answers depend on other model parameters excluding $\pmb{\psi}$ and $\pmb{\lambda}$?

More formally, recall that the set of BE is $\mathcal{B}$. Consider the distance between a pure-strategy profile $\mathbf{s}$ and the set BE as $d(\mathbf{s}, \mathcal{B} ):= \inf \{d(\mathbf{s}, \mathbf{p}^*) | \mathbf{p}^* \in  \mathcal{B} \}$, where $d(\cdot, \cdot)$ is the Euclidean distance between two vectors. Thus, the distance between a pure-strategy profile $\mathbf{s}$ and a set of BE is the infimum of the distances between $\mathbf{s}$ and strategy profiles in the set of BE.

Recall that $\mathcal{B}$ depends on the model parameters $\Gamma$ and $\mathcal{B}(\Gamma)$ explicitly reflects this dependence. To simplify further expression, I define the concept of limiting behavioural equilibrium as follows.

\begin{definition}
For a pure strategy profile $\mathbf{s}$, if $\lim_{\Gamma \to \Gamma_0} d(\mathbf{s}, \mathcal{B}(\Gamma) ) = 0$, then it is said that $\mathbf{s}$ is a limiting behavioural equilibrium (or limiting BE) as $\Gamma \to \Gamma_0$. The number of limiting BE as $\Gamma \to \Gamma_0$ is then the number of pure-strategy profiles that are a limiting BE as $\Gamma \to \Gamma_0$.
\end{definition}

\begin{proposition}
\label{pro: eta}
Suppose Assumptions \ref{ass: payoff-basic}-\ref{ass:finite} hold. Then for any given network $G$, as $\pmb{\psi} \to \pmb{0}$ or $\pmb{\lambda} \to \pmb{\infty}$ (or both), we have that\\
(1) If $wz>xy$, then any strict pure-strategy Nash equilibrium $\mathbf{s}^*$ is a limiting BE. Moreover, the number of limiting BE is (weakly) decreasing with $\pmb{\eta}$.\\
(2) If $wz< xy$, then any limiting BE is a Nash equilibrium. Moreover, the number of limiting BE is (weakly) increasing with $\pmb{\eta}$.
\end{proposition}

The proof is in Appendix. 

Proposition \ref{pro: eta} says that when agents have very long memory or very high accuracy, the set of possible long-term action profiles depends on the payoff matrix and the parameter $\pmb{\eta}$ --- the weight on unselected action when updating attractions. Specifically, if $wz > xy$ (mis-coordination payoffs are not both very negative), then any strict NE, loosely speaking, can be supported in the behavioural model. That is, the long-term action pattern of the dynamical game can be ``close" to any Nash equilibrium starting with some initial conditions. Also, the number of possible long-term actions is (weakly) decreasing with $\pmb{\eta}$, meaning that if people pay less attention to forgone payoffs, then the possibility of long-term action patterns becomes richer.

The result is the opposite when $wz < xy$ (mis-coordination results in very negative payoffs), in which case any long-term action pattern in this behavioural model must be a Nash equilibrium. In particular, some Nash equilibrium may not be supported by the behavioural model. Also, the number of possible long-term actions is (weakly) increasing with $\pmb{\eta}$, suggesting that if people pay more attention to unselected choices, then the possibility of long-term action patterns becomes richer.

The intuition behind these is that when an agent pays little attention to unselected option (low $\eta_i$), he only focuses on the payoff he got without comparing it to the payoff he could have gained had he chose the other strategy. When $wz > xy$, the value of mis-coordination is not very negative (recall that $x<w$ and $y<z$), so agents could still get a ``satisfactory" payoff when they did not make a best response, which indicates that some non-Nash strategy profile may be supported as a long-term action profile. Also, the less agents pay attention to forgone payoffs, the richer the possible long-term action profile could be. On the other hand, if $wz <xy$ meaning large payoff loss in the case of mis-coordination, then people are not satisfied with the results of some partial coordination profile if they pay little attention to unselected option, which induces a strict requirement for a profile to be a stationary state. 

Let me summarise some simple corollaries from Proposition \ref{pro: eta}.

\begin{corollary}
\label{cor1}
Suppose Assumptions \ref{ass: payoff-basic}-\ref{ass:finite} hold. Then as $\pmb{\psi} \to \pmb{0}$ or $\pmb{\lambda} \to \pmb{\infty}$, \\
(1) Suppose $\pmb{\eta}=\pmb{0}$. Let $\mathbf{s}$ be a pure-strategy profile in which $s_i=0$ for any $i \in N_C$ and $s_i=1$ for any $i \in N_D := N \setminus N_C$. Then, $\mathbf{s}$ is a limiting BE if and only if $N_C$ is strictly $r_2$-cohesive and $N_D$ is strictly $r_1$-cohesive, where $r_2=\frac{-y}{z-y}$ and $r_1=\frac{-x}{w-x}$.\\
(2) If $wz>xy$, then the number of $BE$ is (weakly) larger than the number of $NE$.\\
(3) The number of $BE$ is at least $2$.
\end{corollary}

The proof is in Appendix.

Corollary \ref{cor1}-(1) says that if people do not consider forgone payoffs when updating attractions (reinforcement learning) and have high memory and high accuracy, then the set of pure strategy profiles that can be played in the long run is described by the concept of network cohesiveness as in the NE case (\cite{morris2000contagion}), but the threshold $r_1$ and $r_2$ are generally different from those in NE.

Corollary \ref{cor1}-(2) and -(3) consider the number of behavioural equilibria and serve as a comparison with Proposition \ref{pro2} showing how that number differs across the two extreme cases of $\pmb{\psi}$ and $\pmb{\lambda}$. That is, (1) when $\pmb{\psi}$ is large (forgetful) or $\pmb{\lambda}$ is small (inaccurate), then all agents will be in favour or the risk-dominant option in the long run regardless of their initial attractions (Proposition \ref{pro2}). (2), when $\pmb{\psi}$ is small (retentive) or $\pmb{\lambda}$ is large (accurate), then the number of long-term action patterns can be larger than that of NE (Corollary \ref{cor1}-(2)). In particular, it is at least two (Corollary \ref{cor1}-(3)), corresponding to all agents coordinating on either action.

The number of BE is greater when $\pmb{\psi} \to \pmb{0}$ or $\pmb{\lambda} \to \pmb{\infty}$ than when $\pmb{\psi} \to \pmb{\infty}$ or $\pmb{\lambda} \to \pmb{0}$ since when $\pmb{\psi}$ decreases and $\pmb{\lambda}$ increases, people's initial attractions tend to be more ``important" due to increased memory and decision accuracy. The implication is that any pure-strategy profile that is an equilibrium in the static model can be played as a long-term action profile in the dynamic game by ``choosing" the initial attractions consistent with that pure-strategy profile. On the other hand, under low memory and low decision accuracy, people's mixed action profile induced by their attractions may not be strong enough to support some rather ``risky" pure-strategy equilibrium, driving the population to the risk-dominant action.

A question is then what the properties of BE are if $\pmb{\psi}$ and $\pmb{\lambda}$ are between the two extreme cases $\pmb{0}$ and $\pmb{\infty}$. Following the above intuition, one may ask whether the number of BE is decreasing with $\pmb{\psi} $ and increasing with $\pmb{\lambda}$. The answer is negative as summarised in the next Proposition. In particular, it is possible that the number of BE under some values of $\pmb{\psi} $ and $\pmb{\lambda} $ is greater than that when $\pmb{\psi}  \to \pmb{0}$ and $\pmb{\lambda}  \to \pmb{\infty}$.

\begin{proposition}
\label{pro4}
Suppose Assumptions \ref{ass: payoff-basic}-\ref{ass:finite} hold. The number of BE need not be monotonic in $\pmb{\psi}$ or in $\pmb{\lambda}$.
\end{proposition}

\begin{proof}
Show by finding a counter example as follows. 

Suppose the payoff matrix is as in Table \ref{table-exp} and the network structure is as in Figure \ref{exp-net}. Note that one can check that there are only two pure-strategy Nash equilibria of this game: (1) all agents play action $C$, and (2) all agents play action $D$.
\begin{table}[h!]
\centering
\caption{a symmetric coordination game}
\label{table-exp}
 \begin{tabular}{|c |c | c |} 
 \hline
  & $C$ & $D$ \\ [0.5ex]
 \hline
 $C$ & 3 & -5 \\ [0.5ex]  \hline
 $D$ & 0 & 2 \\
  [1ex] 
 \hline
 \end{tabular}
\end{table}

\begin{figure}[H]
  \centering
  \caption{network structure}
  \label{exp-net}
    \includegraphics[width=0.6\textwidth]{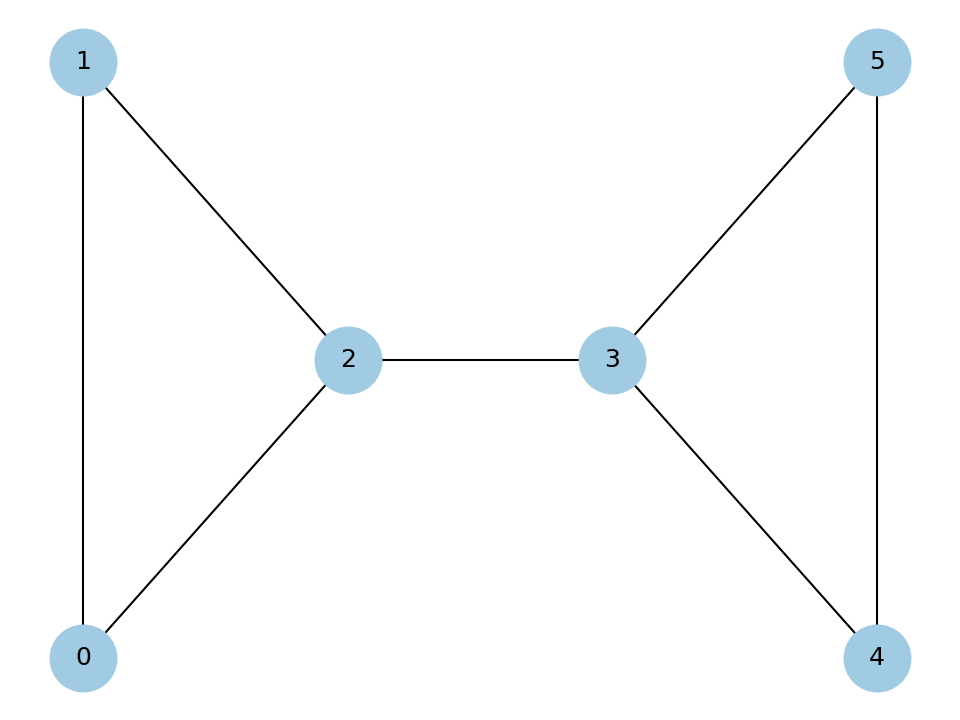}\label{exp-net}
\end{figure}

Suppose all agents conduct belief-based learning (i.e., fix $\pmb{\eta} = \mathbf{1} = [1, 1, 1, 1, 1, 1]$), then consider the following cases.
\\
\par\noindent\textit{Case 1:} When $\pmb{\psi}$ is small and $\pmb{\lambda}$ is large, there are two equilibria. For example, suppose $\pmb{\psi}^1 = [0.5, 0.5, 0.5, 0.5, 0.5, 0.5]$ and $\pmb{\lambda}^1 = [2, 2, 2, 2, 2, 2]$, then there are exactly two BE:
$$ BE_1 (\pmb{\psi}^1, \pmb{\lambda}^1) : \mathbf{p}^* \approx [0, 0, 0, 0, 0, 0] $$ $$ BE_2 (\pmb{\psi}^1, \pmb{\lambda}^1): \mathbf{p}^* \approx [1, 1, 1, 1, 1, 1] $$
\par\noindent\textit{Case 2:} 
When $\pmb{\psi}$ is large enough or $\pmb{\lambda}$ is small enough, then there is a unique BE (in which all agents play the risk-dominant option with probability larger than $0.5$). For example, suppose $\pmb{\psi}^2 = [15, 15, 15, 15, 15, 15]$ and $\pmb{\lambda}^2 = [2, 2, 2, 2, 2, 2]$, then the unique BE is:
$$ BE (\pmb{\psi}^2, \pmb{\lambda}^2) :  \mathbf{p}^* \approx [0.82, 0.82, 0.90, 0.90, 0.82, 0.82] $$
\par\noindent\textit{Case 3:} when $\pmb{\psi}^3 = [0.5, 0.5,  0.5, 0.5, 0.5, 0.5]$ and $\pmb{\lambda}^3 = [0.06, 0.06, 0.06, 0.06, 0.06, 0.06]$, then the unique BE is: $$ BE (\pmb{\psi}^3, \pmb{\lambda}^3) :  \mathbf{p}^* \approx [0.78, 0.78, 0.86, 0.86, 0.78, 0.78] $$
\par\noindent\textit{Case 4:}
When $\pmb{\psi}^4 = [0.5, 0.5, 15, 0.5, 0.5, 0.5]$ and $\pmb{\lambda}^4 = [2, 2, 2, 2, 2, 2]$, then there are three BE:
$$ BE_1 (\pmb{\psi}^4, \pmb{\lambda}^4) : \mathbf{p}^* \approx [0, 0, 0.23, 0, 0, 0] $$
$$ BE_2 (\pmb{\psi}^4, \pmb{\lambda}^4): \mathbf{p}^* \approx [1, 1, 0.94, 1, 1, 1] $$
$$ BE_3 (\pmb{\psi}^4, \pmb{\lambda}^4): \mathbf{p}^* \approx [1, 1, 0.82, 0.04, 0, 0] $$.
\par\noindent\textit{Case 5:}
Also, when $\pmb{\psi}^5 = [0.5, 0.5, 0.5, 0.5, 0.5, 0.5]$ and $\pmb{\lambda}^5 = [2, 2, 0.06, 2, 2, 2]$, then there are three BE:

$$ BE_1 (\pmb{\psi}^5, \pmb{\lambda}^5 ) : \mathbf{p}^* \approx [0, 0, 0.25, 0, 0, 0] $$
$$ BE_2 (\pmb{\psi}^5, \pmb{\lambda}^5 ): \mathbf{p}^* \approx [1, 1, 0.93, 1, 1, 1] $$
$$ BE_3 (\pmb{\psi}^5, \pmb{\lambda}^5 ): \mathbf{p}^* \approx [1, 1, 0.79, 0.01, 0, 0] $$.

To summarise, $\pmb{\psi}^2 \geq \pmb{\psi}^4 \geq \pmb{\psi}^1$, $\pmb{\lambda}^2 =\pmb{\lambda}^4 =\pmb{\lambda}^1$, and $\pmb{\eta}^2 =\pmb{\eta}^4 =\pmb{\eta}^1$, but the number of BE is 1, 3, 2, under $\pmb{\psi}^2$, $\pmb{\psi}^4$, and $\pmb{\psi}^1$, respectively. Also, $\pmb{\lambda}^3 \leq \pmb{\lambda}^5 \leq \pmb{\lambda}^1$, $\pmb{\psi}^3=\pmb{\psi}^5=\pmb{\psi}^1$, and $\pmb{\eta}^3 =\pmb{\eta}^5 =\pmb{\eta}^1$, but the number of BE is 1, 3, 2 under $\pmb{\lambda}^3$, $\pmb{\lambda}^5$, ad $ \pmb{\lambda}^1$, respectively. This shows that the number of BE is not necessarily monotonic in $\pmb{\psi}$ or in $\pmb{\lambda}$.

\end{proof}

This example shows that when $\pmb{\psi}$ and $\pmb{\lambda}$ are between the two extreme cases $\pmb{0}$ and $\pmb{\infty}$, the number of $BE$ can be larger than that under $\pmb{\psi} \to \pmb{0}$ or $\pmb{\lambda} \to \pmb{\infty}$. In particular, when people conduct belief-based learning, the number of $BE$ can be larger than that of Nash equilibria. The key reason is that under intermediate level of decision accuracy and memory, agents' actions can be noisy enough so more possibilities of sustainable action profiles are opened and accurate enough so not all attractions will go toward the risk-dominant option.

To have a better understanding of this, consider again the example in the above proof. The Nash equilibrium does not support diverse actions among players: it is easy to check that every agent in the network requires all its neighbours choosing $C$ in order for the agent in question to choose $C$. However, the actions of agents in the behavioural model are not pure strategies. In particular, agent $2$ in Case \textit{4} and \textit{5} has high $\pmb{\psi}$ and low $\pmb{\lambda}$, respectively, which means that her choice tends to be noisier than others. For example, in the third BE of Case \textit{4}, agent 2 plays action $D$ with a probability of $82\%$ and action $C$ with a probability of $18 \%$, and it is this $18 \%$ probability of playing action $C$ that has a huge implication for agent $3$ and therefore also agents $4$ and $5$: due to agent 2's noises in choosing $D$, agent $3$ now has on average approximately $2.18$ neighbours playing action $D$ as opposed to $2$ when restricting to pure actions, which makes agent $3$'s optimal response action $C$. Also, since the overall decision accuracy is accurate enough, it is sustainable for agents $3$, $4$, and $5$ to favour the efficient option.

So far the focus is on all possible long-term action profiles without explicitly studying which of them (if there are multiple one) the dynamical network game will approach given people's initial predispositions, which will be studied next.

\section{Influence Analysis}
\label{sec:influ}
This section considers the influences of agents with different behavioural features and network positions. Some questions are: whose initial predispositions are important to the actions of the whole population? How does the influence of an agent depend on her and others' behavioural features and the network structure? 

An environment convenient for studying these influence questions is where the two choices are symmetric as in Table \ref{sym matrix} since there is no prior difference between the two options and what choice(s) the population plays in the long run may be attributed to agents' initial preferences.

More concretely, suppose in the network, some group of people $N_C(0)$ initially favour $C$ while others $N_D(0) =N \setminus N_C(0)$ initially favour $D$ and the average initial attraction difference is zero, meaning that people on average are indifferent between these two choices. Then which choice, if any, the network will coordinate on in the long run depends on which of the two groups is on average more influential. Note that the analysis for the general case of asymmetric choices can be accomplished similarly by re-centering agents' average initial attraction differences $\mathbf{q}(0)$, so restricting to the symmetric action case is without loss of generality.\footnote{To see this, note that if the two choices are asymmetric, then there must exist one fixed point $\mathbf{q}^*$ of the dynamical system (\ref{big q}) such that agents favour the risk-dominant option under $\mathbf{q}^*$ and that $\mathbf{q}^*$ becomes unstable when $\pmb{\psi}$ is small and $\pmb{\lambda}$ is large. This $\mathbf{q}^*$ would play the same role as the origin $\mathbf{0}$ does in the case of symmetric choices described in this section.}

\begin{table}[h!]
\centering
\caption{Coordination game with symmetric choices}
%\label{table1}
 \label{sym matrix}
 \begin{tabular}{|c |c | c |} 
 \hline
  & $C$ & $D$ \\ [0.5ex]
 \hline
 $C$ & $h, h$ & $l, l$ \\ [0.5ex]  \hline
 $D$ & $l, l$ & $h, h$ \\
  [1ex] 
 \hline
 \end{tabular}
\end{table}

That is, $\Pi_{00}=\Pi_{11}=h$ and $\Pi_{01}=\Pi_{10}=l$: coordination on either choice produces payoff $h$, while mis-coordination results in payoff $l$. It is assumed that $h>l$.
\\
\par\noindent\textit{Analytical Results}---The dynamical system \ref{big q} is now $\forall i$,
\begin{align}
\begin{split}
\label{dyn-sym}
 \dot{q}_{i}= &-\psi_i q_{i} + (p_{i}+\eta_i(1-p_{i})) \sum_{j \in N_i} \left(p_{j} h + (1-p_{j})l \right) \\ &
  - (1-p_{i}+\eta_i p_{i}) \sum_{j \in N_i}   \left(p_{j} l + (1-p_{j})h \right)
  \end{split}
\end{align}

It is clear that $\mathbf{q}^*=\mathbf{0}$ is a fixed point of the dynamical system \ref{dyn-sym}, which corresponds to the strategy profile $\mathbf{p}^*=\mathbf{\frac{1}{2}}$: all agents play the two choices with equal probability. Whether it is asymptotically stable depends on the model parameters through the Jacobian matrix $J(\mathbf{q}^*)$ evaluated at $\mathbf{q}^*=\mathbf{0}$, which has the form 
\begin{align}
\label{jaco}
    \begin{split}
      &  J_{ii}=-\psi_i + \frac{(1-\eta_i) d_i (h+l) \lambda_i}{4}, \ \ \forall i \\
      &  J_{ij}= \frac{G_{ij}(1+\eta_i) (h-l) \lambda_j}{4}\ , \ \ \forall i, j, j \neq i
    \end{split}
\end{align}
Standard theory from dynamical systems tells us that

\begin{claim}
The state $\mathbf{q}^*=\mathbf{0}$ is asymptotically stable (or equivalently $\mathbf{p}^*=\mathbf{\frac{1}{2}}$ is a BE) if all the eigenvalues of the Jacobian matrix \ref{jaco} have negative real parts.
\end{claim}

\begin{example}\textbf{2:}
Suppose all agents conduct belief-based learning (i.e., $\pmb{\eta}= \mathbf{1}$), and $\psi_i =\psi$ and $\lambda_i =\lambda$, $\forall i$, then the Jacobian matrix is equal to $-\psi I + \frac{\lambda (h-l)}{2} G$. Thus, the state $\mathbf{q}^*=\mathbf{0}$ is asymptotically stable if the largest eigenvalue of the graph $G$ is smaller than $\frac{2\psi}{ \lambda (h-l)}$. This suggests that (under belief-based learning) the state in which all agents are indifferent between the two choices tends to be stable when agents are forgetful and inaccurate, when the payoff difference between coordination and mis-coordination is small, and when the network is ``sparse". 
\end{example}

To analyse the influence question, I focus on the case in which the state $\mathbf{q}^*=\mathbf{0}$ is not stable (otherwise starting with a state near $\mathbf{q}^*=\mathbf{0}$, the action pattern would converge to $\mathbf{q}^*=\mathbf{0}$, so there would be no difference in influences across agents in the sense of influencing the long-term action profile). 

The analysis relies on the linear approximation near the state $\mathbf{q}^*=\mathbf{0}$:
\begin{align}
\label{approx_linear}
    \begin{split}
        \dot{\mathbf{q}} & = F(\mathbf{q})  \approx F(\mathbf{0}) + J(\mathbf{0}) \mathbf{q} \\
        & =J(\mathbf{0}) \mathbf{q}  = B K B^{-1} \mathbf{q}
    \end{split}
\end{align}
where $BKB^{-1}$ is the eigendecomposition of the Jacobian matrix $J(\pmb{0})$. Then the nonlinear dynamical system is approximated by a linear dynamical system near the state $\mathbf{q}^*=\pmb{0}$. Suppose the eigenvalues of the matrix $J(\mathbf{0})$ are $\kappa_1$, $\kappa_2$,..., $\kappa_n$ with $\kappa_1$ being the largest eigenvalue,\footnote{It is assumed in the paper that the $n$ eigenvalues are distinct, a property which holds generically.} and its associated right eigenvectors are $v_1$, $v_2$, ..., $v_n$, and associated left eigenvectors are $u_1$, $u_2$, ..., $u_n$. Then it is well known that the linearised dynamic system (denoted by $\hat{\mathbf{q}}(t)$) has the following solution:
\begin{align}
\label{linear_sol}
    \begin{split}
       \hat{\mathbf{q}}(t)= \sum_{r=1}^{n} u_r^T \mathbf{q}(0) e^{\kappa_r t} v_r
    \end{split}
\end{align}

That is, the solution to the linear dynamical system is a linear combination of eigenvectors of $J(\pmb{0})$ with the weight on each eigendirection being $u_r^T \mathbf{q}(0) e^{\kappa_r t}$. The direction that matters most in the long run is the one having the largest eigenvalue. In principle, when $t$ goes to infinity, the eigendirection associated with the largest eigenvalue dominates the trajectory $\hat{\mathbf{q}}(t)$, but since the dynamical system is nonlinear and the linear system is just a local approximation around $\mathbf{q}^*=\pmb{0}$, the argument in terms of the direction of the action pattern is valid only when $\mathbf{q}$ is sufficiently close to the origin.

More formally, the next Proposition shows that when agents' initial attraction differences are sufficiently close to $\mathbf{q}^*=\pmb{0}$, then the long-term action pattern can be determined by a weighted average of initial attraction differences $\{q_1(0), q_2(0),...,q_n(0) \}$ with the weights being the first left eigenvector. The condition for the eigenvector approximation to be valid is that agents are nearly indifferent between the two actions at the beginning. This might happen when agents have limited knowledge about the payoffs of the two choices, have little understanding of the game, or have little information about which actions others will be choosing.

\begin{proposition}
\label{pro: influence}
Let the payoff structure be in Table \ref{sym matrix}. Suppose that the state $\mathbf{q^*}=\pmb{0}$ is unstable and the right eigenvector associated with the largest eigenvalue have all elements positive. Then for any given $\delta>0$, there exists an $\epsilon >0$ such that for any $\mathbf{q}(0)$ with $ || \mathbf{q}(0) || \in (0, \epsilon) $, it must be that $\frac{u_1^T \mathbf{q}(0)}{||\mathbf{q}(0) ||}  >\delta$ implies that the long-run action pattern has $\mathbf{q^*}>\mathbf{0}$ (all agents favour action $D$), while $\frac{u_1^T \mathbf{q}(0)}{||\mathbf{q}(0) ||} < -\delta$ implies that the long-run action pattern has $\mathbf{q^*}<\mathbf{0}$ (all agents favour action $C$).
\end{proposition}

The proof is in Appendix.

The idea of Proposition \ref{pro: influence} is that under the assumption that the first right eigenvector has all elements positive, $u_1^T \mathbf{q}(0) $ determines whether the linearised dynamical system will go toward $\hat{\mathbf{q}}>\mathbf{0}$ or $\hat{\mathbf{q}}<\mathbf{0}$ in the long run. Specifically, if $u_1^T \mathbf{q}(0) > 0$ ($<0$), then the linearised dynamical system will have $\hat{\mathbf{q}}(t)>\mathbf{0}$ ($<\mathbf{0}$) for all $t > \overline{t}$ for some $\overline{t}$. Note that the linear approximation works well only when the attraction vector is close to the origin, while the linearised dynamical system $\hat{\mathbf{q}}(t)$ itself has all elements positive or negative only after some time $\overline{t}$. Since the magnitude of $\mathbf{q}(0)$ does not impact the value of the threshold $\overline{t}$, then if $\mathbf{q}(0)$ is sufficiently close to the origin, $\hat{\mathbf{q}}(t)$ will still be sufficiently close to the origin at time $\overline{t}$. This indicates that $\mathbf{q}(t)$ will also have either $\mathbf{q}(t)>\mathbf{0}$ or $< \mathbf{0}$ after time $\overline{t}$, which represent all agents favouring action $D$ and $C$ in the long run, respectively.

For convenience, denote the influence vector by $\pmb{\xi} := u_1$ ($u_1$ is the first left eigenvector of the Jacobian matrix $J(\mathbf{0})$). Loosely speaking, Proposition \ref{pro: influence} says that as long as people are nearly indifferent between the two choices initially, then when $\sum_i^n \xi_i q_i(0) $ is strictly positive (negative), the population will coordinate on action $D$ ($C$). In this sense, $\pmb{\xi}$ summarises the influence of each agent. The initial predispositions of agents who have a higher (lower) $\xi_i$ have a higher (lower) weight on determining which action the population will play in the long term.

To study how the influences depend on the behavioural parameters and network structures, note that the influence of agent $i$ satisfies:
\begin{align}
\label{lefteigen}
    \begin{split}
        \xi_i = \kappa_1 \sum_{j=1}^{n}  \xi_j J_{ji}
    \end{split}
\end{align}

where recall that
\begin{align}
\label{jaco2}
    \begin{split}
      &  J_{ii}=-\psi_i + \frac{(1-\eta_i) d_i (h+l) \lambda_i}{4}, \ \ \forall i \\
      &  J_{ji}= \frac{G_{ji}(1+\eta_j) (h-l) \lambda_i}{4}\ , \ \ \forall i, j, j \neq i
    \end{split}
\end{align}

Thus, the following simple comparative statics can be seen directly from (\ref{lefteigen}) and (\ref{jaco2}).

\begin{claim}
Let $\pmb{\xi}$ be the left eigenvector associated with the largest eigenvalue of $J(\pmb{0})$. $\pmb{\xi}$ represents agents' influences in the sense of Proposition \ref{pro: influence}, and \\
(1) $\xi_i$ is decreasing with $\psi_i$; \\
(2) if $(h+l)>0$ or if $\eta_i=1$, then $\xi_i$ is increasing with $\lambda_i$.
\end{claim}

In addition to these, as what eigenvector centrality commonly reflects in other contexts, the agents who connect to highly influential people tend to have high influence themselves. The intuition in this paper's setting is that an agent's initial attraction can impact her neighbours' attractions and if her neighbours are highly influential, then the impacts can propagate to the agent's neighbours' neighbours. Also, loosely speaking, if the neighbours of an agent have high $\eta_j$ (i.e., not ignoring unselected choices), then, ceteris paribus, the agent in question will also tend to have a high influence.

Note that the left eigenvector as a measure of importance also appears in the DeGroot learning model (e.g., \cite{degroot1974reaching} , \cite{golub2010naive}, \cite{golub2017learning}, etc.), where the importance of agents is summarised by the eigenvector centrality of the belief updating matrix. Here the Jacobian matrix evaluated at $\mathbf{q^*}=\pmb{0}$ plays a similar role as the updating matrix does in the learning model in that both describe how the states (attractions or beliefs) of agents influence each other. There are, however, three main differences: (1) the impact of one's state on others' now depends on the behavioural features, namely $\pmb{\lambda}$ and $\pmb{\eta}$, which interact with the information about the network and payoff structures; (2) the diagonal element of the updating matrix in the learning model reflects the weight an agent puts on herself, while here it reflects attraction depreciation and reinforcement learning effects; (3) the DeGroot learning model is a linear system while the model in this paper is nonlinear so the prediction made by the left eigenvector only works when people are nearly indifferent between the two choices at the beginning.

The following example shows both how the aforementioned influence vector works and how the influences of agents may depend on their behavioural features and network positions.
\\
\par\begin{example}\textbf{3:}
Consider a simple 3-node star network shown in Figure \ref{inf-net}, where agent $0$ is the center and has two degrees, while agents $1$ and $2$ have degree one. Suppose the payoff function is that coordination on either option results in payoff $2$ while mis-coordination produces payoff $-1$. That is, $h=2$ and $l=-1$. The behavioural parameters are as follows: $\pmb{\psi}=[1, 1, \frac{1}{2}]$, $\pmb{\lambda}=[\frac{1}{2}, \frac{1}{2}, 1]$, and $\pmb{\eta}=[\frac{1}{2},\frac{1}{2},\frac{1}{2}]$.

\begin{figure}[H]
  \centering
  \caption{network structure}
  \label{inf-net}
    \includegraphics[width=0.5\textwidth]{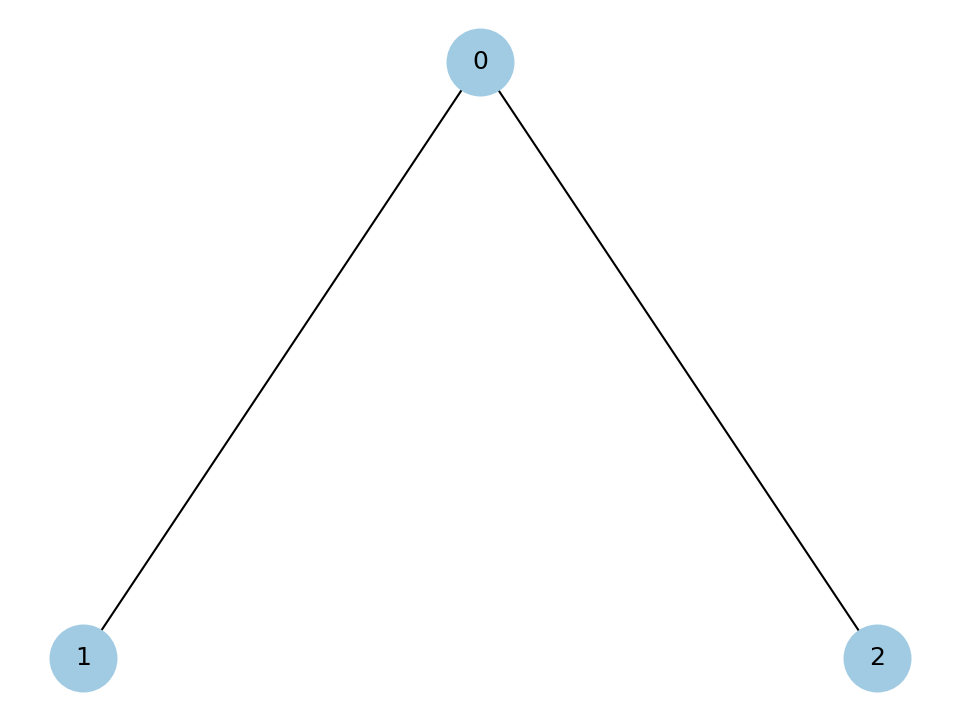}\label{inf-net}
\end{figure}

The solution to the linearised dynamical system is:
\begin{align}
\label{linear_sol}
    \begin{split}
        \hat{\mathbf{q}}(t) = & [0.32, 0.15, 0.53] \cdot \mathbf{q}(0)  e^{0.31 t}  [1.21, 0.55, 1.00]^T\\
         & +  u_2^T \mathbf{q}(0) e^{-0.76 t}  v_2 +  u_3^T \mathbf{q}(0)  e^{-1.74 t} v_3
    \end{split}
\end{align}
A principal component approximation is thus given by
\begin{align}
\label{linear_sol}
    \begin{split}
       \hat{\mathbf{q}}(t) \approx \left( 0.32 q_1(0) + 0.15 q_2(0) + 0.53 q_3(0)  \right)   e^{0.31 t}  [1.21, 0.55, 1.00]^T 
    \end{split}
\end{align}

As the first right eigenvector has all elements positive, the sign of $0.32 q_1(0) + 0.15 q_2(0) + 0.53 q_3(0) $ determines the sign of $\hat{\mathbf{q}}(t)$ when $t$ is large. That is, if the weighted summation of initial attraction differences, $0.32 q_1(0) + 0.15 q_2(0) + 0.53 q_3(0) $, is positive (negative), then $\hat{\mathbf{q}}(t)$ will have all elements positive (negative) when $t$ is large, which also implies that the long-run action pattern is that the population coordinating on action $D$ ($C$) if $\mathbf{q}(0)$ is sufficiently close to $\mathbf{0}$ in the sense of Proposition \ref{pro: influence}.

The influence vector is thus $[0.32, 0.15, 0.53]$. Agent $0$ has a higher influence ($0.32$) than agent $1$ does ($0.15$) despite their behavioural parameters being identical. This is because agent $0$ has degree $2$ while agent $1$ has degree $1$. Also, agent $2$ has the highest influence ($0.53$) as it has the lowest depreciation rate and highest decision accuracy among the three. 

To give some impression of the approximation, suppose that the initial attraction differences of the three agents are $[0.1, 0.1, -0.18]$. Note that the simple summation of the three initial attractions are positive, but the weighted summation is negative as $0.32 q_1(0) + 0.15 q_2(0) + 0.53 q_3(0) \approx -0.049$. Thus, the importance vector predicts that all attractions will become negative, which means that they will favour action $C$ in the long run.

\begin{figure}[H]
  \centering
  \caption{Approximations}
    \includegraphics[width=0.8\textwidth]{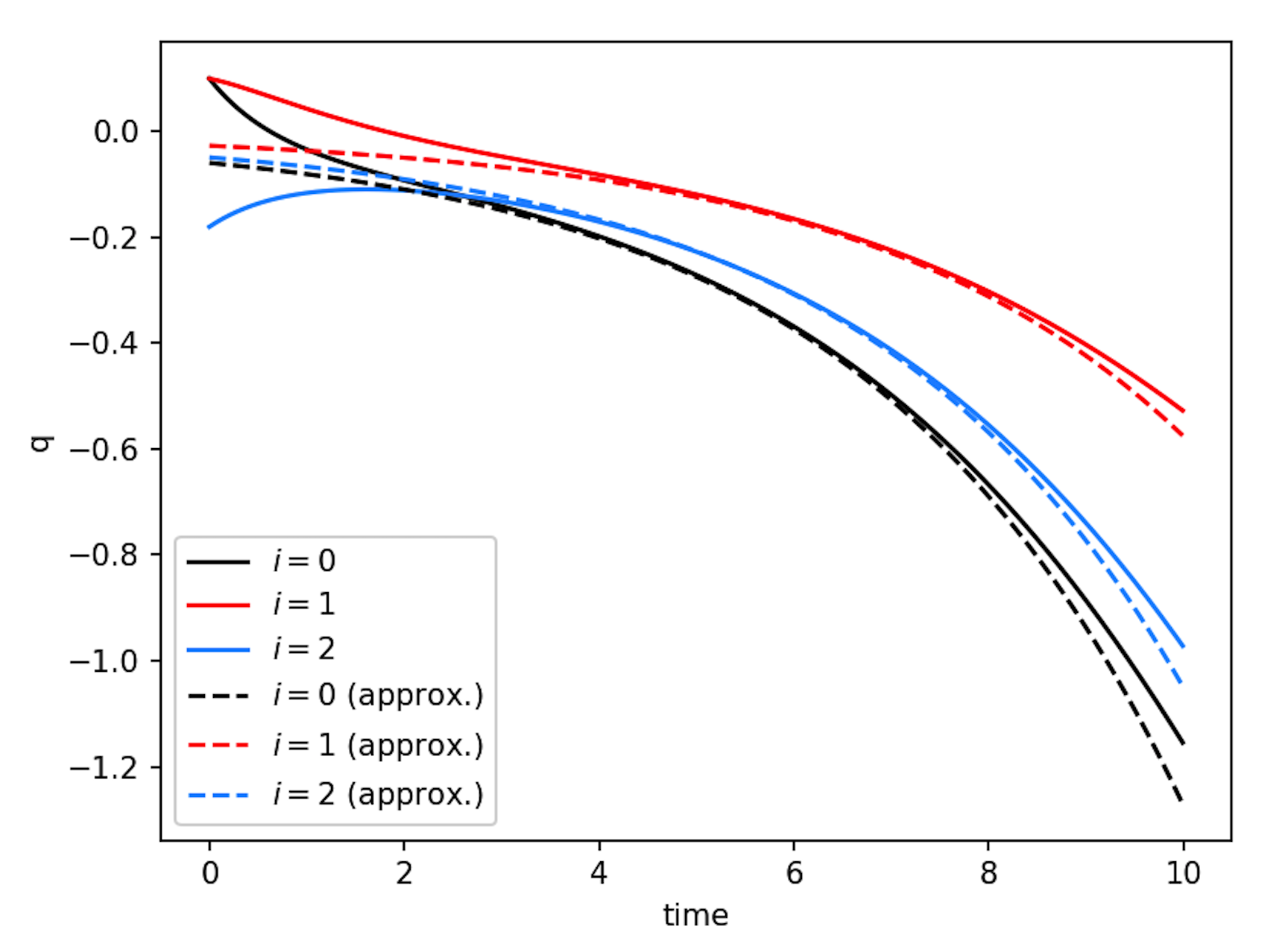}\label{exp_app}
    \caption*{\small Notes: The solid curves represent the dynamics of $q_i(t)$ for the three agents ($i=0, 1, 2$) for $t \in [0, 10]$. The dashed curves represent the corresponding approximations of $q_i(t)$ for each agent, where the approximations are the first principal component of the linearised dynamical system (i.e., $u_1^T \mathbf{q}(0) e^{\kappa_1 t} v_1$).}
\end{figure}

Figure \ref{exp_app} shows that this is indeed the case. Note that when $t$ is small ($t<2$), the approximation is not very good due to the ignorance of other eigencomponents. The approximation is most accurate in intermediate values of $t$. These are the periods when the first principal component is a good approximation of the linear dynamical system $\hat{\mathbf{q}}(t)$ and when the linear dynamical system is a good approximation of the original nonlinear dynamical system $\mathbf{q}(t)$. When $t$ is large, the approximation must diverge from the original system, but this is irrelevant to the analysis of influences as all attraction differences have had the same sign across agents and the population will then remain at the state in which all favour the choice corresponding to that sign.

\end{example}

\par\noindent\textit{Numerical Calculations} --- So far I have shown analytically that the influence of agents can be characterised by the left eigenvector associated with the largest eigenvalue of a Jacobian matrix under some conditions, the most important of which is that people should be nearly indifferent between the two choices initially. This part continues to explore the influence question by numerical calculations with the following two main aims. The first is to show how well the left eigenvector can predict the long-term action profile in practice (when the initial attractions are not arbitrarily close to indifference). Also, the left eigenvector may be hard to interpret in practice, so the second objective is to show how the influences depend on some statistics of the distributions of agents' initial attractions, behavioural parameters, and network positions.

In the numerical calculations, the network is fixed to be the Erdos-Renyi graph (\cite{erdos1960evolution}), $G(n, p)$, with $n=100$ and $p=0.1$. The network will be redrawn if it is not connected.\footnote{As mentioned earlier, if a network is unconnected, then we can analyse each connected component separately, so only connected networks are considered.} To fix idea, the payoff structure is constrained to one of the following two with each accounting for roughly $50 \%$ of the simulations: (1) $\Pi_{00}=\Pi_{11}=2$, $\Pi_{01}=\Pi_{10}=-1$ and (2) $\Pi_{00}=\Pi_{11}=1$, $\Pi_{01}=\Pi_{10}=-2$.\footnote{First, note that the focus of the simulations is influence analysis so only symmetric payoff structures are considered as in the analytical part. Also, the magnitude of the payoff structure needs to be normalised because the effect of changing it can be equivalent to adjusting the magnitude of other parameters, e.g., $\pmb{\lambda}$.} The parameters $\pmb{\psi }$, $\pmb{ \lambda}$, and $\pmb{\eta}$ are all assumed to follow uniform distributions: $\psi_i \sim U(\underline{\psi}, \overline{\psi})$, $\lambda_i \sim U(\underline{\lambda}, \overline{\lambda})$, and $\eta_i \sim U (\underline{\eta}, \overline{\eta})$, where $\underline{\psi}$, 
$\overline{\psi}$, $\underline{\lambda}$, and $\overline{\lambda}$ are themselves drawn uniformly from $[0.1, 10]$ and $\underline{\eta}$ and $\overline{\eta}$ are drawn from $[0,1]$.\footnote{Note that the types of distributions used per se are not important, at least to a first order. The purpose of having those distributions is to explore a relatively large range of parameter values.}

The initial condition $\mathbf{q}(0)$ is assumed to follow a normal distribution with mean zero and standard deviation $\sigma_q$ where $\sigma_q$ is drawn from $[0.01,1]$. Recall that the mean is set to zero so the ex ante probability of coordinating on choice $C$ is equal to that of coordinating on choice $D$ before knowing further information about the distributions of $\mathbf{q}(0)$ and other parameters across agents. In each simulation, all the parameters are drawn independently of each other and across agents.

I conduct $100, 000$ simulations. In $98.77 \%$ of the simulations, the population has either all agents favour $D$ (i.e., $q^*_i>0$, $\forall i$) or all agents favour $C$ (i.e., $q^*_i<0$, $\forall i$). Among them, $50.13 \%$ of the simulations favours action $D$ while the remaining favours action $C$. Thus, the fraction of coordinating on either choice is roughly $50\%$, which is of no surprising due to the symmetry in the data generating process.

The interest lies in how the influences of agents represented by the first left eigenvector $\pmb{\xi}$ described in the analytical part can predict whether the population will coordinate on action $C$ or on action $D$. Recall that the prediction is obtained by the weighted sum of agents' initial attractions, $\pmb{\xi}^T \mathbf{q}(0)$: if this number is positive (negative), then the prediction is that the population will favour action $D$ ($C$).

Restricting to the $99.73 \%$ of the observations where the first eigenvalue is strictly positive,\footnote{Recall that a negative first eigenvalue indicates that the action dynamics will converge to the profile where all agents are indifferent between choices $C$ and $D$, which is not of interest in the influence analysis. Also, in all the simulations the first right eigenvector has all elements positive.} the prediction generated by the aforementioned eigenvector importance has prediction accuracy $85.20 \%$ and $85.15 \%$ for predicting coordination on choices $D$ and $C$, respectively. That is, if the weighted sum $\pmb{\xi}^T \mathbf{q}(0)$ is positive (negative), then in $85.20 \%$ ($85.15 \%$) of the time, the population will indeed favour action $D$ ($C$) in the long run as predicted. Note that these numbers are significantly larger than $50 \%$, the accuracy would have been obtained from random predictions.

Figure \ref{eigen} plots the prediction accuracy for action $D$ as a function of the standard deviation of the initial attraction differences, $\sigma(\mathbf{q}(0))$. The result shows that the prediction accuracy is falling with that standard deviation, starting with as high as $98\%$ when $\sigma(\mathbf{q}(0))$ is below $0.1$, to close to $80\%$ when $\sigma(\mathbf{q}(0))$ approaches $3$. These outcomes are consistent with the fact that the eigenvector prediction is most valid when people are nearly indifferent between the two choices initially, but they also indicate that even when there is some variation in agents' initial attractions ($\sigma(\mathbf{q}(0)) \approx 3$), the explanatory power is still significant ($80\% \gg 50\%$).
\\
\begin{figure}[H]
  \centering \caption{Prediction accuracy vs. $\sigma(\mathbf{q}(0))$}
    \includegraphics[width=0.7\textwidth]{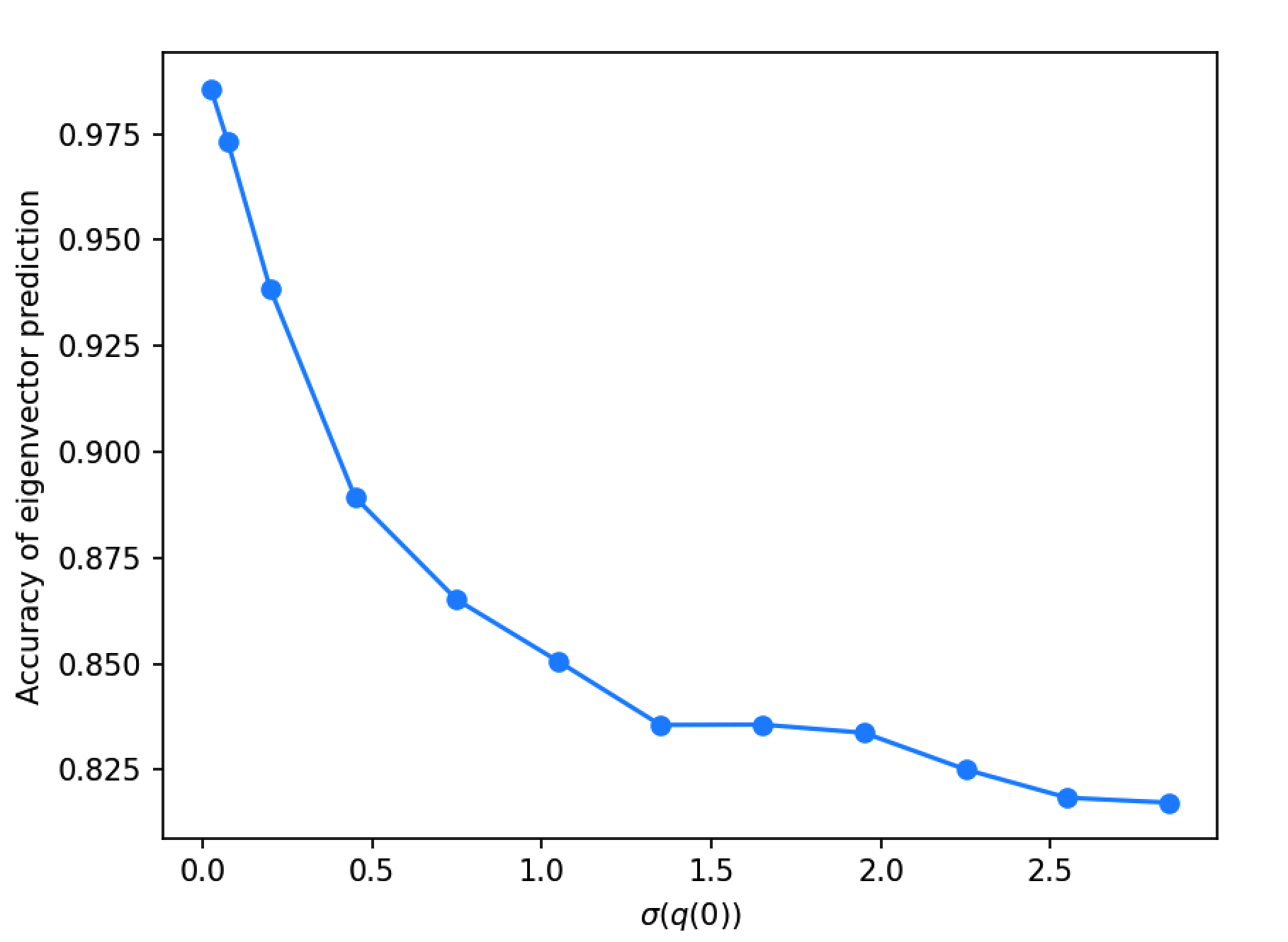} 
    \label{eigen}
    \caption*{\small Notes: This Figure plots the prediction accuracy for coordinating on action $D$ as a function of the standard deviation of the initial attraction differences, $\sigma(\mathbf{q}(0))$. Each dot represents the average prediction accuracy conditional on $\sigma(\mathbf{q}(0))$ being in an interval around the value in the x-axis. }
\end{figure}

Having shown the predictive power of the influence vector, I then explicitly examine how the influences of agents depend on their behavioural features and network positions, as the left eigenvector of the aforementioned Jacobian matrix, which summarises all the information of the model primitives, may not be very interpretable. 

\begin{figure}[H]
  \centering
  \caption{Frequency of playing action $D$ vs. statistics}
  \begin{minipage}[b]{0.49\textwidth}
    \includegraphics[width=\textwidth]{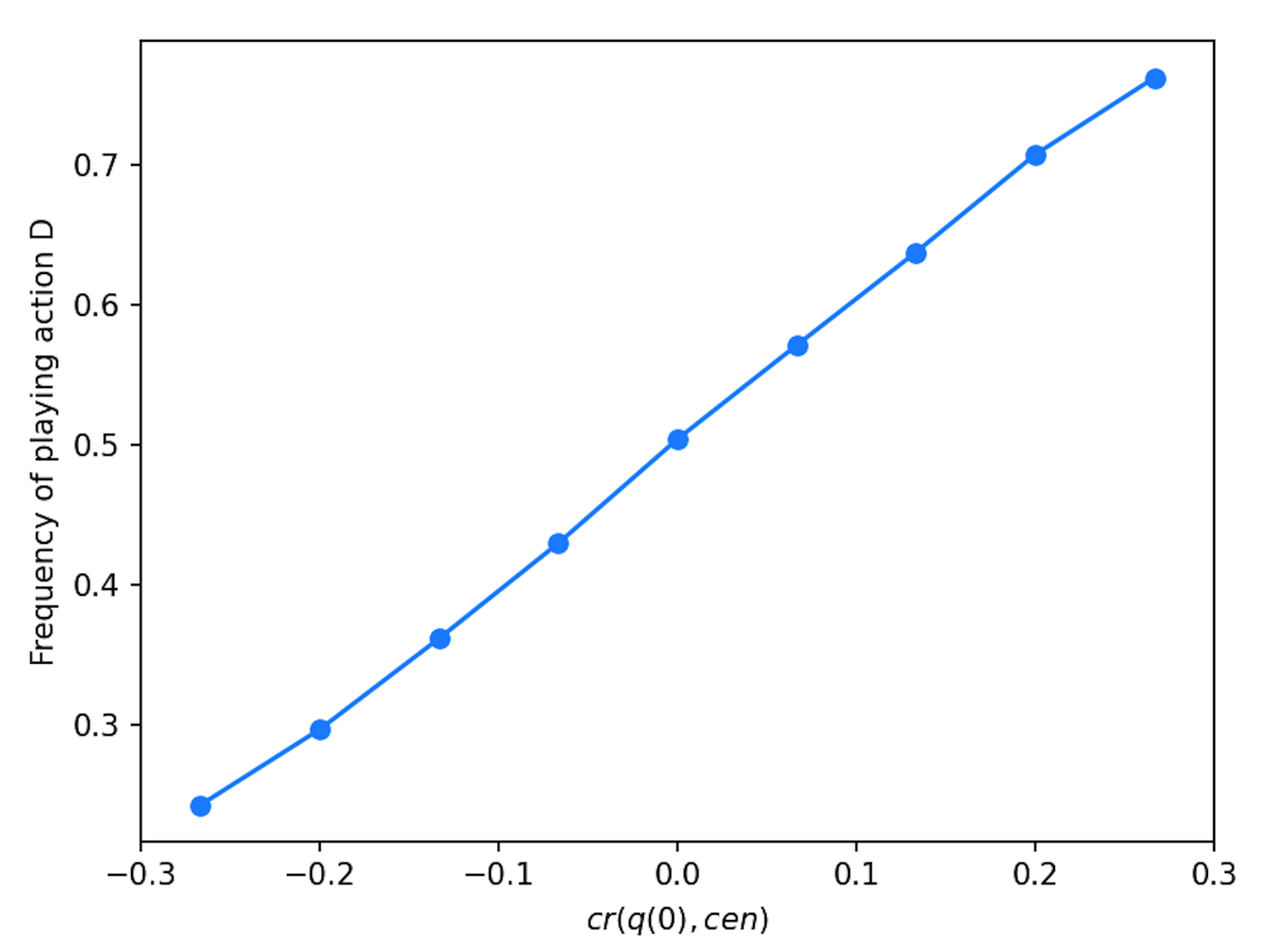}\subcaption{Frequency vs. $cr(\mathbf{q}(0), cen)$}
  \end{minipage}
  \hfill
  \begin{minipage}[b]{0.49\textwidth}
    \includegraphics[width=\textwidth]{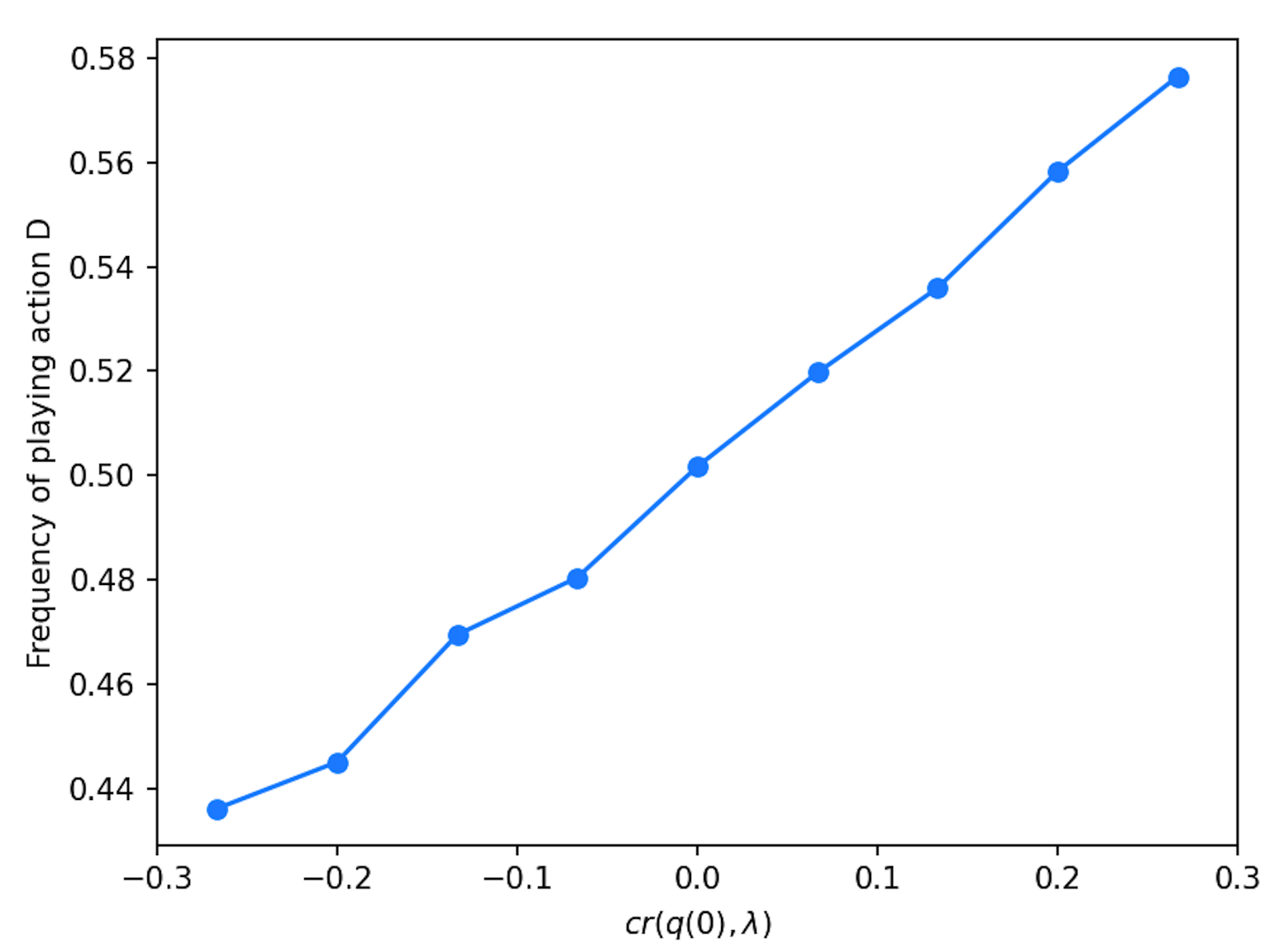}  \subcaption{Frequency vs. $cr(\mathbf{q}(0), \pmb{\lambda})$}
  \end{minipage}
  \label{part_hete}
  \caption*{\small Notes: Notations --- $cr(\cdot, \cdot)$ refers to the correlation coefficient between two variables. $cen$ refers to agents' network eigenvector centrality. (a) The frequency of the whole population favouring action $D$ conditional on the correlation coefficient between initial attraction differences and eigenvector centrality in the network, $cr(\mathbf{q}(0), cen)$. (b) The frequency of the whole population favouring action $D$ conditional on the correlation coefficient between initial attraction differences and decision accuracy, $cr(\mathbf{q}(0), \pmb{\lambda})$.}
\end{figure}

The partial dependence plot in Figure \ref{part_hete}-(a) shows that the frequency of the whole population favouring action $D$ is increasing with the correlation coefficient\footnote{Recall that in the data generating process, each variable including the initial condition is independent of each other, but there will still exist some small (empirical) correlation between the variables in the generated data due to randomness.} between $\mathbf{q}(0)$ and network eigenvector centrality,\footnote{This eigenvector centrality is the eigenvector corresponding to the largest eigenvalue of the adjacency matrix of the network and should not be confused with the influence vector represented by the left eigenvector corresponding to the largest eigenvalue of the Jacobian matrix.} while Figure \ref{part_hete}-(b) shows that that frequency is increasing with the correlation coefficient between initial attraction differences and decision accuracy, although to a less extent. Recall that $\mathbf{q}(0)$ is people's attraction differences between choice $D$ and choice $C$, so a positive correlation between $\mathbf{q}(0)$ and, say, $\pmb{\lambda}$ means that those favouring action $D$ are on average more accurate than those favouring action $C$. Thus, Figure \ref{part_hete}-(a) and -(b) indicate that, consistent with intuition, people with a higher network centrality and higher accuracy level have a larger influence on which outcome the population will favour as compared to their respective counterparts.

\begin{figure}[H]
  \centering
  \caption{Frequency vs. $cr(\mathbf{q}(0), \pmb{\lambda})$ for different $\mu(\pmb{\lambda})$}
    \includegraphics[width=0.6\textwidth]{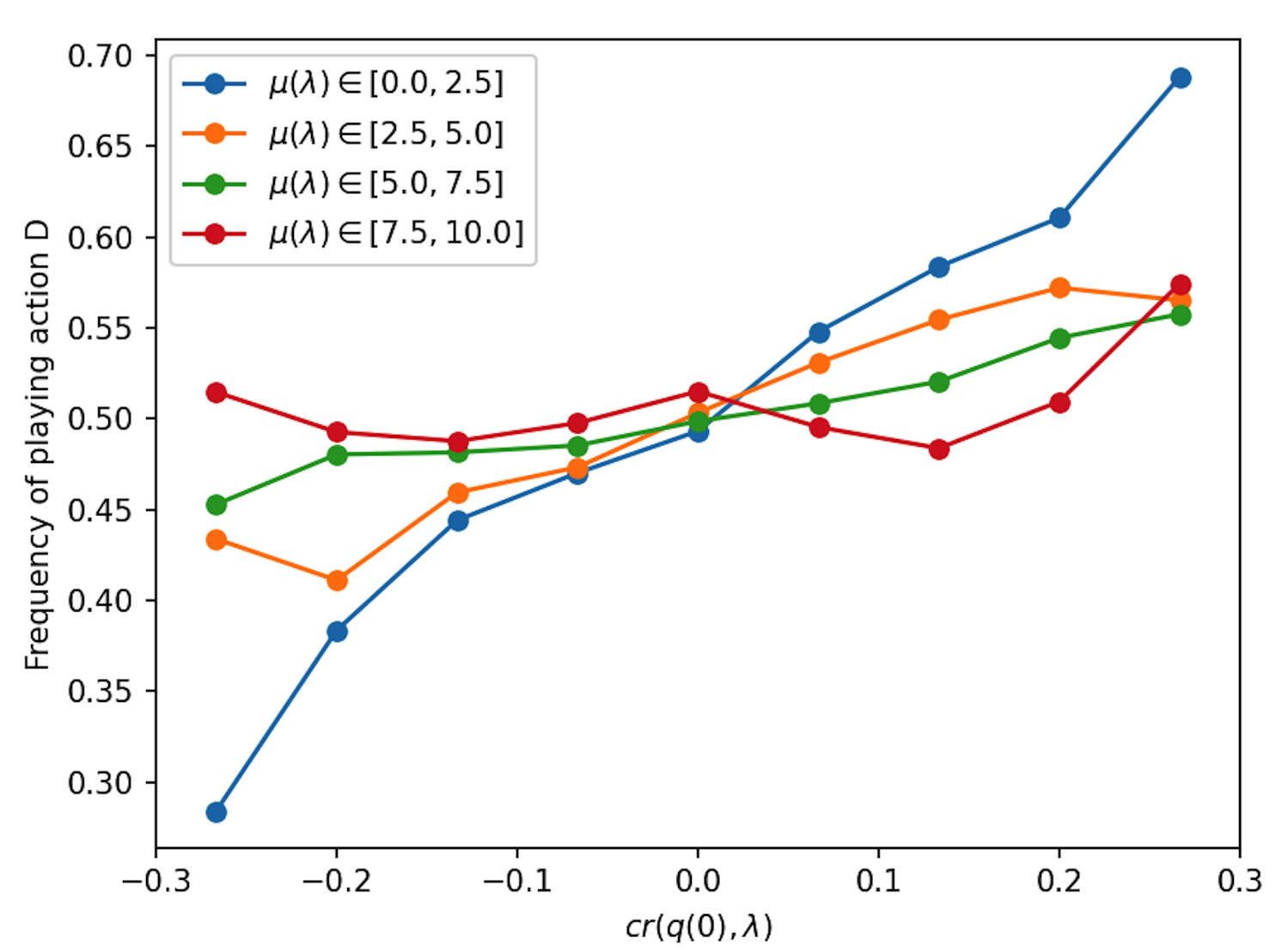}
  \label{int_freq}
  \caption*{\small Notes: This Figure plots the frequency of the whole population favouring action $D$ conditional on the correlation coefficient between initial attraction differences and decision accuracy, for each different value of average decision accuracy. This plot shows the interaction effects between $cr(\mathbf{q}(0), \pmb{\lambda})$ and $\mu(\pmb{\lambda})$ in predicting which choice the population will favour.}
\end{figure}

Figure \ref{int_freq} shows that the impact of $cr(\mathbf{q}(0), \pmb{\lambda})$ on the coordination outcome is larger when the average decision accuracy is small as compared to when it is high. In particular, people having a large $\lambda$ seem to make a significant difference only when the overall $\lambda$ is below $2.5$ (the blue curve). The intuition is that when agents are overall not very accurate, then relatively highly accurate agents can lead the inaccurate people to choose the action they like. In contrast, when all agents have high accuracy, then those with an even higher decision accuracy than the average no longer make a significant difference.

\section{Model Extensions}
\label{sec: extension}
In this section, I discuss two model extensions. Section \ref{sub:time} reduces the assumption of time-invariant behavioural features and studies the case where the decision accuracy is increasing with time. Section \ref{sub:imitate} considers the situation in which agents have a high attention to the best payoff they experienced.

\subsection{Time-Variant Behavioural Parameters}
\label{sub:time}
It might be plausible that people's behavioural features are changing over the process of game play. This section relaxes the time-invariant assumption on the behavioural parameters in the baseline model. Specifically, I assume that people's accuracy level, $\pmb{\lambda}$, is increasing with time and goes to infinity as time goes to infinity. In other words, people may be inaccurate in decision making initially but the inaccuracy level drops as more experience is accumulated and vanishes eventually. For simplicity, parameters $\pmb{\psi}$ and $\pmb{\eta}$ are still assumed to be time-invariant.\footnote{It can be shown that the case of decreasing $\pmb{\psi}$ produces similar outcomes to the case of increasing $\pmb{\lambda}$.}

\begin{assumption}
\label{ass: behaviour}
Suppose $\pmb{\psi} > \pmb{0}$ and $\pmb{\eta} \in [\pmb{0}, \pmb{1}]$ are time-invariant and $\pmb{\lambda}(t)$ is increasing with $t$ and $\lim_{t \to \infty} \pmb{\lambda}(t) = \pmb{\infty}$.
\end{assumption}

Many results in the baseline model can be extended to the case of vanishing noises. For example, the action pattern will generally converge as in Proposition \ref{pro1}, and the relationship between limiting BE and NE as well as the monotonic relationship involving limiting BE and $\pmb{\eta}$ in the sense of Proposition \ref{pro: eta} which considers time-invariant $\pmb{\lambda} \to \pmb{\infty}$ also hold in the extended model as $\lim_{t \to \infty} \pmb{\lambda}(t) = \pmb{\infty}$.

Also, recall that Proposition \ref{pro2} studies the case where the (time-invariant) $\pmb{\lambda}$ is small. This condition will ultimately not hold when $t$ is large as $\lim_{t \to \infty} \pmb{\lambda}(t) = \pmb{\infty}$, but Proposition \ref{pro2} still have implication for the extended model, which is that starting with any given initial attractions, if people's accuracy level is sufficiently small in early periods, then the long-term action profile must be that everyone favours the risk-dominant option, as summarised in the following Corollary.

\begin{corollary}
\label{cro4}
Suppose that Assumptions \ref{ass:network}, \ref{ass: risk-effi}, and \ref{ass: behaviour} hold. Then for any given $\mathbf{q}(0)$, there exists a vector $\underline{\pmb{\lambda}}$ and a positive number $\underline{T}$ such that if $\pmb{\lambda}(t) < \underline{\pmb{\lambda}}$ for any $t < \underline{T}$, then when $t \to \infty$, all players favour the risk-dominant option.
\end{corollary}

The proof is analogous to that of Proposition \ref{pro2} so is omitted. This Corollary shows that given any (finite) initial attractions, if agents are inaccurate for a sufficiently long period of time, then the long-term action pattern must be coordination on the risk-dominant option even if all agents are in favour of the efficient option initially. Same as Proposition \ref{pro2}, this is because inaccuracy (noises) in action erodes the overall tendency of playing the efficient choice in the population and this tendency will gradually drive the whole population to favour the risk-dominant option.

In addition to the analogous results to the baseline model, the introduction of time-variant $\pmb{\lambda}$ has implications for the influence questions in terms of which long-term action patterns the population will play. As different agents may have different increasing rate of accuracy level, the time-varying feature of $\pmb{\lambda}$ provides another dimension of heterogeneity that can play a role in equilibrium selection, as shown in the next Proposition.

\begin{proposition}
Suppose that Assumptions \ref{ass:network}, \ref{ass: risk-effi}, and \ref{ass: behaviour} hold. For any given $\mathbf{q}(0)$, let $N_C(0)$ and $N_D(0)$ be players that initially favour $C$ and $D$, respectively. Then, \\
(1) For any network $G$, there exist some behavioural parameters such that all agents' actions converge to the risk-dominant action $D$. \\
(2) If the group $N_D(0)$ does not contain a subset with network cohesiveness (weakly) larger than $\frac{2(z-x)}{z-x+w-y}$, then there exist some behavioural parameters such that all agents' actions converge to the efficient action $C$.
\end{proposition}

 \begin{proof}
The first statement follows from Corollary \ref{cro4}.

To show the second statement, suppose initially all agents in $N_D(0)$ ($N_C(0)$) have sufficiently low (large) $\lambda$, then an agent in $N_D(0)$ connected with the component $N_C(0)$ updates her attraction level according to:
\begin{align}
\begin{split}
\label{pro4-1}
 \dot{q}_{i} \approx & -\psi_i q_{i} + (\frac{1+\eta_i}{2})  \left( \frac{1}{2}(w+x-y-z)n^D_i+ (x-z) (d_i-n^D_i ) \right)
  \end{split}
\end{align}
where $n^D_i$ is the number of $i$'s neighbours favouring $D$. Thus, for large $\psi_i$, $\dot{q}_i <0$ if $\frac{n^D_i}{d_i} < \frac{2(z-x)}{z-x-y+w}$. Therefore, if the agents in $N_D(0)$ who have connection with $N_C(0)$ have sufficiently large $\psi$ and small $\lambda$ and the $N_C(0)$ group has sufficiently small $\psi$ and large $\lambda$, then the group in $N_D(0)$ adjacent to $N_C(0)$ will transform to favouring $C$ within some time $\tau$. Repeating this process and constructing the behavioural parameters for the updated $C$ group and $D$ group in the same way, it follows that all agents will favour action $C$ (the efficient choice) within some time.
 \end{proof}

Note that the threshold $\frac{2(z-x)}{z-x+w-y}$ is twice that of the original cascade model (\cite{morris2000contagion}), indicating that when agents have heterogeneous behavioural features, the requirements for the whole population to transform into a new state can be much looser than in the original model which only considers best responses. Agents with a higher accuracy in decision-making and a longer memory tend to have a larger influence on the action dynamics of the whole network than their respective counterparts do. Highly accurate and retentive agents are rather stubborn to their initial attractions, which might lead their neighbours to switch to the actions they favour. The implication is that if there are systematic differences in behavioural features across the group favouring different options, then the threshold of cascade can be different as compared to the case that only considers best responses.

\subsection{Reinforcing the Best}
\label{sub:imitate}
Some papers (e.g., \citet{robson1996efficient} and \citet{alos2008contagion}) have shown that the efficient option of the coordination game survives equilibrium selection when agents conduct ``imitating-the-best" behavioural rule, which is in contrast to most other papers that select the risk-dominant option (e.g., \citet{ellison1993learning}) in the unique stochastically stable state. To show how these outcomes may be reconciled in this paper, I consider a behavioural feature where players put a high weight on the best payoff they experienced. This may be captured by introducing two additional parameters $\gamma$ and $\underline{\pi}$ in the EWA framework. 
\begin{align}
\begin{split}
\label{attraction updates}
&\dot{a}_{i,1} =-\psi_i a_{i,1} + (p_{i}+\eta_i(1-p_{i})) \sum_{j \in N_i} \left(p_{j} (w - \underline{\pi}) ^{\gamma_i} + (1-p_{j})(x- \underline{\pi})^{\gamma_{i}}  \right) \\
& \dot{a}_{i,0} =-\psi_i a_{i,0} + (1-p_{i}+\eta_i p_{i}) \sum_{j \in N_i}   \left(p_{j} (y- \underline{\pi})^{\gamma_i} + (1-p_{j})(z- \underline{\pi})^{\gamma_i} \right)
\end{split}
\end{align}
where $\underline{\pi}$ is some value smaller than the lowest possible payoff. If $\gamma_i>1$, then agents are risk-loving as they have a tendency of being attracted to the action that resulted in the highest payoff. Note that in contrast, $\gamma_i<1$ reflects risk aversion as in \cite{fudenberg2019predicting}.

It can be shown that there exists a $\hat{\gamma} >1$ such that when $\gamma_i > \hat{\gamma}$ for all $i$, then the action profile with all players favouring the efficient option has a larger basin of attraction than that with all players favouring the risk-dominant option. In particular, if agents are forgetful and inaccurate, then the long-term action profile must be such that all agents favour the efficient option regardless of their initial attractions. The intuition is that if people put a high weight on the highest payoff they experienced, then a small tendency of playing the efficient option in the population may be augmented as it could generate the highest possible return.

This is loosely consistent with the outcomes indicating that imitating-the-best behavioural rule results in equilibrium selection favouring the efficient rather than the risk-dominant option (e.g., \cite{robson1996efficient} and \cite{alos2008contagion}). Note that this result lies in that people put a high weight on the ``best" outcome they witnessed instead of in the ``imitation" behaviour. Indeed, if people perform the ``imitating-the-average" rule, then the risk-dominant option (instead of the efficient option) will still be selected.

\section{Concluding Remarks}
\label{sec:con}
This paper studies network coordination games with bounded-rational agents who conduct experience-weighted attraction learning. The main distinction lies in the consideration of multiple different behavioural features and the heterogeneities in them across agents. The long-term action profile of the game is a high-dimensional function of the network structure, payoff matrix, and all agents' behavioural features and initial attractions. As high-dimensional mappings generally lack tractability and interpretability, I explore what patterns can be obtained from it. 

I show that the set of possible long-term action profiles can be largely different when the behavioural features vary. When agents are sufficiently forgetful and inaccurate, they will favour the risk-dominant option in the long run regardless of their initial predispositions. When agents are sufficiently retentive and accurate, the set of possible long-term action profiles can be richer than that of Nash equilibrium, depending on the payoff matrix and on the attention people pay to unselected choices. Possible long-term action profiles can be richer under intermediate level of forgetfulness and decision accuracy than under the above two extreme cases, meaning that the number of them needs not be monotonic in forgetfulness or in decision accuracy. In terms of which long-term action profile will be played when there are multiple one, I show that it can be determined by a weighted sum of agents' initial attractions provided that agents’ initial attractions are sufficiently close to some neutral level, with the weights being the principal left eigenvector of a Jacobian matrix. This eigenvector reflects agents' influences which summarises the information about the distribution of behavioural features as well as network and payoff structure.

There are multiple directions of future research. First, it is natural to study the properties of long-term action profiles of other network games and under other types of behavioural features in the context of the EWA model. The reinforcing-the-best behavioural rule and the consideration of aspiration levels briefly described in Section \ref{sec: extension} and Appendix, respectively, are some examples. With regard to games, the network coordination games studied in the paper exhibit strategic complements across agents. It might be interesting to examine the properties of the EWA dynamics in general games with strategic complements and strategic substitutes.

Moreover, from a numerical point of view, more systematic data analysis of the simulation outcomes might provide further insights into the question. For example, the relationship among variables from the simulations in this paper could be learned using a graphical neural network as a surrogate model so that the patterns between long-term action profile and the whole network inputs can be detected without manually defining and measuring some statistics. For example, it might be interesting to examine whether ML methods can predict agents' influences better than the eigenvector discussed in the paper does. 

\section{Appendix}
\label{sec:app}

\par\noindent\textbf{Aspiration-based reinforcement learning:}
\label{sub: aspiration}

Consider that each agent has an aspiration level of utility $\hat{u}_i$ --- an agent ``aspires"  that she could get an utility $\hat{u}_i$ and she will compare the payoff she obtained (or could have obtained) from playing an action with that aspiration level, and if it is higher (lower) than the aspiration level, then the attraction of that action tends to increase (decrease). This can be represented by
\begin{align}
\begin{split}
\label{attraction updates}
&\dot{a}_{i,1} =-\psi_i a_{i,1} + (p_{i}+\eta_i(1-p_{i})) \sum_{j \in N_i} \left(p_{j} w + (1-p_{j})x -\hat{u}_i \right) \\
& \dot{a}_{i,0} =-\psi_i a_{i,0} + (1-p_{i}+\eta_i p_{i}) \sum_{j \in N_i}   \left(p_{j} y + (1-p_{j})z -\hat{u}_i  \right)
\end{split}
\end{align}
Thus, the analysis of this dynamical game is equivalent to that with the stage payoff:
\begin{table}[H]
\centering
\caption{a canonical coordination game}
\label{table1}
 \begin{tabular}{|c |c | c |} 
 \hline
  & $C$ & $D$ \\ [0.5ex]
 \hline
 $C$ & $z-\hat{u}_i$ & $y-\hat{u}_i$ \\ [0.5ex]  \hline
 $D$ & $x-\hat{u}_i$ & $w -\hat{u}_i$ \\
  [1ex] 
 \hline
 \end{tabular}
\end{table}
where $\hat{u}_i$ is $i$'s aspiration level. Thus, the baseline analysis is equivalently to the case in which all agents' aspirations are zero. Also, if all agents have the same (possibly non-zero) aspiration level, then the analysis is equivalently to a parallel transformation of the payoff matrix ($\Pi=\hat{\Pi}-\hat{u}$). This indicates that the baseline analysis assuming $w>0$ and $z>0$ is without loss of generality in the above sense (as long as the aspiration levels are homogeneous across agents). The inclusion of heterogeneous and time-varying aspiration level is beyond the scope of this research and is a potential future direction.\footnote{See e.g., \citet{karandikar1998evolving} and \cite{borgers2000naive} who consider evolving aspiration levels.}
\\
\par\noindent\textbf{Proof of Proposition \ref{pro: eta}}:

(1) Denote $m_i$ as the number of $i$'s neighbours choosing action $D$. In a NE, $s_i^*=D$ indicates that \begin{align}
\label{p1-3}
 \frac{m_i^*}{d_i} > \frac{ z -x}{w-x+ z- y} :=r
\end{align}
and 
$s_i^*=C$ indicates that \begin{align}
\label{p1-4}
 \frac{d_i-m_i^*}{d_i} > \frac{ w -y}{ w- x + z-y} := 1-r
\end{align} Thus, suppose that there is a NE such that group $N_C^*$ chooses action $C$ and group $N_D^*$ chooses action $D$ where $N_C^* \cup N_D^* =N$, then as in \cite{morris2000contagion}, $N_D^*$ forms a $r$-cohesive set while $N_C^*$ forms a $(1-r)$-cohesive set.

Now consider BE. Note that in a fixed point, an individual $i$ must have $p_i=p_i^*$ such that
\begin{align}
\label{exp: pure}
   -\frac{\psi_i}{\lambda_i} \ln \left( \frac{1}{p_i}-1 \right) =  (p_{i}+\eta_i(1-p_{i})) u(p_i=1, p_{-i})- (1-p_{i}+\eta_i p_{i})  u(p_i=0, p_{-i})
\end{align}
given $p_{-i}=p^*_{-i}$.

Consider a vector $\mathbf{p}$ very close to $\mathbf{p^*}$, for an agent $i$ who has $s_i^*=1$ in the NE,
\begin{align}
\label{exp: ne1}
\begin{split}
  & (p_{i}+\eta_i(1-p_{i})) u(p_i=1, p_{-i})- (1-p_{i}+\eta_i p_{i})  u(p_i=0, p_{-i}) \\
  = \ &   u(p_i=1, p_{-i}) - \eta_i  u(p_i=0, p_{-i}) + O (1-p_i) \\
  =\ & u(p^*_i=1, p^*_{-i}) - \eta_i  u(p^*_i=0, p^*_{-i})+ O( || \mathbf{p}-\mathbf{p}^* ||)  \\
  =\ & m^*_i w +(d_i -m_i^*) x - \eta_i (m^*_i y + z ( d_i-m^*_i))+ O( || \mathbf{p}-\mathbf{p}^* ||)  \\
  =\ &  d_i \left[(w-x-\eta_i y +\eta_i z)\frac{m_i^*}{d_i} +x-\eta_i z \right]  + O( || \mathbf{p}-\mathbf{p}^* ||)   \\
  > \ &  d_i \left[(w-x-\eta_i y +\eta_i z)(r+\epsilon_1) +x-\eta_i z \right]  + O( || \mathbf{p}-\mathbf{p}^* ||)  \\
  = \ &  d_i \frac{(1-\eta_i)(wz-xy) }{w+z-x-y}+d_i (w-x-\eta_i y +\eta_i z) \epsilon_1 + O( || \mathbf{p}-\mathbf{p}^* ||) \\
  \geq \ & 0 +d_i (w-x-\eta_i y +\eta_i z) \epsilon_1 + O( || \mathbf{p}-\mathbf{p}^* ||) \\
  > \ & 0
\end{split}
\end{align}
when $\mathbf{p}$ is close to $\mathbf{p}^*$. The $\epsilon_1>0$ is some constant where the strictness of NE is used. The second last line uses the condition that $wz>xy$.
Similarly, one can show that when $\mathbf{p}$ is very close to $\mathbf{p^*}$, for an agent $i$ who has $s_i^*=0$ in the NE,
\begin{align}
\label{exp: ne0}
\begin{split}
  & (p_{i}+\eta_i(1-p_{i})) u(p_i=1, p_{-i})- (1-p_{i}+\eta_i p_{i})  u(p_i=1, p_{-i}) \\
   < \ & - \epsilon_0 + O( || \mathbf{p}-\mathbf{p}^* ||)
\end{split}
\end{align}
for some constant $\epsilon_0>0$. Thus, there exists a neighbour $N(\mathbf{p^*})$ of $\mathbf{p^*}$ such that for any $\mathbf{p} \in N(\mathbf{p^*})$, the right hand side (RHS) of (\ref{exp: pure}) is strictly positive for any $i$ with $s_i^*=1$ and strictly negative for any $i$ with $s_i^*=0$ when $\mathbf{p}\in N(\mathbf{p^*})$.

Define a mapping $\mathbf{p'} := H(\mathbf{p})$ be such that for each $i$, $p'_i : = H_i(\mathbf{p})$ is a solution to equation (\ref{exp: pure}) given $p_{-i}$ and that $p'_i$ is the solution that is the closet to $p^*_i$.\footnote{There must exist at least one solution and may exist multiple solutions.} Construct a $n$-dimensional compact set $\mathfrak{P}=\Pi_i [\underline{p}_i, \overline{p}_i] \subseteq N(\mathbf{p^*})$ be such that for each $i$, $\overline{p}_i=1$ if $s_i^*=1$, and $\underline{p}_i=0$ if $s_i^*=0$. Fix such compact set $\mathfrak{P}$. Then for any $\mathbf{p} \in \mathfrak{P}$, we have that when $\frac{\psi_i}{\lambda_i}$ is sufficiently small, then $1> H_i(\mathbf{p}) > \underline{p}_i$ for any $i$ with $s_i^*=1$ and $\overline{p}_i> H_i(\mathbf{p}) > 0$ for any $i$ with $s_i^*=0$. Thus, $H$ is a mapping from $\mathfrak{P}$ to $\mathfrak{P}$ and since $H_i(\cdot)$ is continuous, applying the Brouwer's fixed-point theorem, there must exists a fixed point\footnote{It can also be easily shown that the fixed point is stable by calculating the limit of the Jacobian matrix. Details are omitted.} in $\mathfrak{P}$. As $\pmb{\psi} \to \mathbf{0}$ or $\pmb{\lambda} \to \pmb{\infty}$, the fixed point $\mathbf{p^*} \in \mathfrak{P}$ converges toward the point $\mathbf{s}^*$ since the LHS of \ref{exp: pure} goes to zero as $\frac{\psi_i}{\lambda_i}$ goes to zero unless $p_i$ approaches $0$ or $1$, which also indicates that $d(\mathbf{p}^*, \mathbf{s}^*) \to 0$, so $\mathbf{s}^*$ is a limiting BE.

Now suppose for a pure-strategy profile $\mathbf{s}$, as $\pmb{\psi} \to \pmb{0}$ or $\pmb{\lambda} \to \pmb{\infty}$, $d(\mathbf{s}, \mathcal{B}) \to 0$ under $\pmb{\eta}''$. This means that as $\pmb{\psi} \to \pmb{0}$ or $\pmb{\lambda} \to \pmb{\infty}$, there always exists a $\mathbf{p}^* \in BE$ for each parameter value such that $p^*_i \to 0$ or $1$ in the limit for any $i$.

For $p^*_i \to 1$, it can be shown that \begin{align}
\label{p1-1}
\begin{split}
     &  w m_i^* +x (d_i-m_i^*) > \eta_i y m_i^* + \eta_i z (d_i -m_i^*) \\
 \Longrightarrow \ & \frac{m_i^*}{d_i} > \frac{\eta_i z -x}{w-x+\eta_i z- \eta_i y} := r_1(\eta_i)
\end{split}
\end{align}

Similarly, for $p^*_i \to 0$, we have that \begin{align}
\begin{split}
\label{p1-2}
 &  y m_i^* +z (d_i-m_i^*) > \eta_i w m_i^* + \eta_i x (d_i -m_i^*) \\
 \Longrightarrow \ & \frac{d_i-m_i^*}{d_i} > \frac{\eta_i w -y}{z-y+\eta_i w- \eta_i x} := r_2(\eta_i)
 \end{split}
\end{align}

Note that 
\begin{align}
\begin{split}
\label{p1-5}
 \frac{\partial r_1( \eta_i)}{\partial  \eta_i} & = \frac{z(w-x+\eta_i z- \eta_i y)-(\eta_i z-x)(z-y)}{(w-x+\eta_i z- \eta_i y)^2} \\
 & = \frac{zw -xy}{(w-x+\eta_i z- \eta_i y)^2}
 \end{split}
\end{align}
and
\begin{align}
\begin{split}
\label{p1-6}
 \frac{\partial  r_2(\eta_i)}{\partial  \eta_i} & = \frac{w(z-y+\eta_i w- \eta_i x)-(\eta_i w-y)(w-x)}{(z-y+\eta_i w- \eta_i x)^2} \\
 & = \frac{zw -xy}{(z-y+\eta_i w- \eta_i x)^2}
 \end{split}
\end{align}
Thus, if $wz>xy$, then both $r_1(\eta_i)$ and $r_2(\eta_i)$ are increasing with $\eta_i$. Suppose $\pmb{\eta}'' \geq \pmb{\eta}'$, then for any $i$, $r_1(\eta''_i) \geq r_1(\eta'_i)$ and $r_2(\eta''_i) \geq r_2(\eta'_i)$. If $\mathbf{p^*} \to \mathbf{s^*}$ under $\pmb{\eta''}$, then ($\ref{p1-1}$) and ($\ref{p1-2}$) hold for $s_i^*=1$ and $s_i^*=0$, respectively. Then using similar steps as of the constructions of the compact set $\mathfrak{P}$ and mapping $H$ described above, it follows that $\mathbf{s}$ is a limiting BE under $\pmb{\eta}'$. This completes the proof of (1). 

The proof of (2) uses the fact that when $wz<xy$, both $r_1(\eta_i)$ and $r_2(\eta_i)$ are decreasing with $\eta_i$ from (\ref{p1-5}) and (\ref{p1-6}), then for any $i$, $r_1(\eta''_i) \leq r_1(\eta'_i)$ and $r_2(\eta''_i) \leq r_2(\eta'_i)$, the remaining steps are analogous.
\\
\par\noindent\textbf{Proof of Corollary \ref{cor1}}:

(1) From the proof of Proposition \ref{pro: eta}, let $\eta_i =0$, we have that 
\begin{align}
    r_1(\eta_i= 0)=\frac{\eta_i z -x}{w-x+\eta_i z- \eta_i y}= \frac{-x}{w-x}
\end{align}
and
\begin{align}
    r_2(\eta_i=0)=\frac{\eta_i w -y}{z-y+\eta_i w- \eta_i x}  =\frac{-y}{z-y}
\end{align}

The remaining steps follow from the arguments in the proof of Proposition \ref{pro: eta}. 

(2) This statement follows from statement (1) of Proposition \ref{pro: eta}.

(3)  From (\ref{p1-1}) and (\ref{p1-2}), $r_1(\eta_i) <1 $ and $r_2(\eta_i)<1$ for any $\eta_i$ based on the assumption that $w, z >0$. Thus, all agents coordinating on either action is a limiting BE as $\pmb{\psi} \to \pmb{0}$ or $\pmb{\lambda} \to \pmb{\infty}$. This in turn means that the number of BE is at least two. 
\\
\par\noindent\textbf{Proof of Proposition \ref{pro: influence}}:

Suppose that the initial condition is $\mathbf{q}(0)=\epsilon \Tilde{\mathbf{q}}(0)$ where $\epsilon >0$ and $||\Tilde{\mathbf{q}}(0)||=1$. Fix $\delta>0$, there exists some $\overline{t}$ independent of $\epsilon$ such that for any initial condition $\epsilon \Tilde{\mathbf{q}}(0)$ with $u_r^T  \Tilde{\mathbf{q}}(0)>\delta$, we have that $\hat{\mathbf{q}}(t) > \mathbf{0}$ whenever $t>\overline{t}$. In particular, $\hat{\mathbf{q}}(2\overline{t})>\pmb{0}$.

As $\mathbf{q}$ approaches $\pmb{0}$, $F(\mathbf{q})=\hat{F}(\mathbf{q})+ o(\mathbf{q}) $, which implies that $\hat{\mathbf{q}}(t) = \mathbf{q}(t) + o(\mathbf{q}(t))$ and $\mathbf{q}(t) = \hat{\mathbf{q}}(t)  + o(\hat{\mathbf{q}}(t))$. Note also when $\epsilon$ is close to zero, then $\mathbf{\hat{q}}(2\overline{t})$ approaches zero. Thus, there exists an $\epsilon$ such that $\mathbf{q}(2\overline{t}) = \hat{\mathbf{q}}(2\overline{t}) + o(\hat{\mathbf{q}}(2\overline{t}) ) >\pmb{0} $.

That is, at time $t=2\overline{t}$, the state is that all agents favour choice $D$ (i.e., $\mathbf{q}(2\overline{t}) > \pmb{0}$), indicating that the dynamical system will converge to the behavioural equilibrium in which all people favour action $D$ as the origin is unstable. The argument is the same for the case of coordination on the action $C$.

\clearpage
\bibliography{ref}

\end{document}